\documentclass[12pt]{amsart}
\usepackage{amssymb}
\usepackage{hyperref}
\usepackage{enumerate}
\usepackage{verbatim}
\usepackage{graphicx}
\usepackage{color}
\usepackage{bbm}
\usepackage{amscd}
\usepackage{footnote}   
\makesavenoteenv{tabular}

\newcommand{\ra}{\rightarrow}
\newcommand{\xra}{\xrightarrow}
\newcommand{\xla}{\xleftarrow}
\newcommand{\RR}{\mathbb{R}}

\newcommand{\mr}{\mathrm}
\newcommand{\dd}{\partial}
\newcommand{\bt}{\bullet}
\newcommand{\til}{\widetilde}
\newcommand{\wh}{\widehat}
\newcommand{\mc}{\mathcal}

\newcommand{\F}{\mathcal{F}}

\newcommand{\FF}{\mathcal{F}}
\newcommand{\hra}{\hookrightarrow}
\newcommand{\bl}{\color{blue}}
\newcommand{\LL}{\mathcal{L}}
\newcommand{\p}{\mathsf{p}}
\newcommand{\ul}{\underline}
\newcommand{\N}{\mathsf{N}}
\newcommand{\HH}{\mathsf{H}}
\newcommand{\G}{\mathbb{G}}
\newcommand{\HHH}{\hat{H}}

\theoremstyle{remark}
\newtheorem{remark}{Remark}[section]
\theoremstyle{plain}
\newtheorem{lemma}[remark]{Lemma}

\newtheorem{thm}[remark]{Theorem}

\theoremstyle{definition}
\newtheorem{definition}[remark]{Definition}

\newtheorem{conjecture}[remark]{Conjecture}

\begin{document}
\title[Fiber products in BV-BFV]{A note on gluing via fiber products in the (classical) BV-BFV formalism
}
\begin{abstract}
In classical field theory, gluing spacetime manifolds along boundary corresponds to taking a fiber product of the corresponding spaces of fields (as differential graded Fr\'echet manifolds) up to homotopy. We construct this homotopy explicitly in several examples in the setting of BV-BFV formalism (Batalin--Vilkovisky formalism with cutting--gluing).
\end{abstract}

\author{Alberto S. Cattaneo}
\address{Institut f\"ur Mathematik, Universit\"at Z\"urich\\
Winterthurerstrasse 190, CH-8057 Z\"urich, Switzerland}  
\email{cattaneo@math.uzh.ch}

 \thanks{A. S. C. acknowledges partial support of SNF Grant No. 200020 192080 and of the Simons Collaboration on Global Categorical Symmetries. This research was (partly) supported by the NCCR SwissMAP, funded by the Swiss National Science Foundation.}

\author{Pavel Mnev}

\address{University of Notre Dame}

\address{St. Petersburg Department of V. A. Steklov Institute of Mathematics of the Russian Academy of Sciences}
\email{pmnev@nd.edu}

\maketitle


\setcounter{tocdepth}{3}
\tableofcontents

\subsection*{Acknowledgments}
We thank Eugene Rabinovich for sharing some ideas related to this work. We are also very much indebted to Domenico Fiorenza and Bruno Vallette for helpful comments and for suggesting (independently) the construction sketched in Section \ref{ss: sketch of proof of Conj 2.2}.

\section{Preliminary definitions: weak equivalences of dg symplectic manifolds}

\subsection{Differential geometric setting}
We start by fixing some definitions. First we recall a (standard) notion of a dg symplectic  manifold (everywhere in the paper ``dg'' stands for ``differential graded''). 
  The subsequent definition of weak equivalence is not (to our knowledge) a part of the standard lore of dg geometry but has been used in the BV-BFV formalism \cite[Definition 2.6.3]{CSS}.
\begin{definition}\label{def: k-symp mfd}
A $k$-symplectic
manifold\footnote{A word of warning: in the literature, the term ``$k$-symplectic'' is sometimes used to indicate that instead of a symplectic $2$-form, one has a closed $k$-form subject to suitable nondegeneracy condition. Here we do not mean that: $k$ for us is the intrinsic (cohomological) degree of the nondegenerate closed $2$-form.} (a more full name is ``degree $k$ dg symplectic manifold'') is a triple $(\F,Q,\omega)$ consisting of a $\mathbb{Z}$-graded supermanifold $\F$ equipped with a degree $+1$ cohomological vector field $Q$ and a degree $k$ symplectic form $\omega\in \Omega^2(\F)_k$ satisfying the compatibility condition 
\begin{equation}\label{def of k-symp mfd: L_Q omega=0}
\LL_Q \omega=0.
\end{equation}

If a $k$-symplectic manifold $(\F,Q,\omega)$ is additionally equipped with a Hamiltonian function $S\in C^\infty(\F)_{k+1}$ for $Q$, i.e., one has $\iota_Q\omega=\delta S$,\footnote{
We use the notation $\delta$ for de Rham differential on $\F$ (which will later be the space of fields in a field theory), and we reserve $d$ for the  de Rham differential on the underlying spacetime manifold $M$ of that field theory. 
} we call the quadruple $(\F,Q,\omega,S)$ a $k$-Hamiltonian manifold.
\end{definition}

The case relevant for the rest of the paper will be $k=-1$ -- the case of the Batalin--Vilkovisky formalism, where $\FF$ will be the space of fields of a field theory, $Q$ the BRST operator and $S$ the action functional. 

A note on terminology: $(-1)$-Hamiltonian manifolds are also known as BV manifolds and $0$-Hamiltonian manifolds -- as BFV manifolds (they arise in Batalin-Fradkin-Vilkovisky construction for homological resolution of coisotropic reductions). Furthermore, $(k-1)$-Hamiltonian manifolds are elsewhere also called BF$^{k}$V manifolds.

\begin{definition}\label{def: WE manifolds}
A weak equivalence of two $k$-symplectic manifolds $(\F,Q,\omega)$
and $(\til\F,\til Q,\til \omega)$ is a pair of dg maps $f\colon \F\ra \til \F$, $g\colon \til \F\ra \F$ such that
\begin{enumerate}[(a)]
\item \label{def: WE manifolds (a)} The pullbacks $f^*\colon C^\infty(\til\F)\ra C^\infty(\F)$ and $g^*\colon C^\infty(\F)\ra C^\infty(\til\F)$ are quasi-isomorphisms of chain complexes, inducing mutually inverse maps in cohomology.
\item \label{def: WE manifolds (b)} One has $f^*\til\omega=\omega+\LL_Q \beta$ and $g^*\omega=\til\omega+\LL_{\til Q} \til\beta$
for some closed 2-forms $\beta\in \Omega^2(\F)_{k-1}$, $\til\beta\in \Omega^2(\til\F)_{k-1}$.
\end{enumerate}
\end{definition}

\begin{lemma} 
\label{lemma 1.3}
If two $k$-Hamiltonian manifolds $(\F,Q,\omega,S)$ and $(\til\F,\til Q,\til \omega, \til S)$ are weakly equivalent as $k$-symplectic manifolds, then the corresponding Hamiltonian functions are related by 
\begin{eqnarray}
f^* \til S &=& S-\frac12 \iota_Q\iota_Q \beta , \label{WE for S 1} \\
g^* S &=& \til S -\frac12 \iota_{\til Q} \iota_{\til Q} \til \beta , \label{WE for S 2}
\end{eqnarray}
where the equalities are modulo locally constant functions if $k=-1$ and on the nose for $k\neq -1$.
\end{lemma}
Lemma \ref{lemma 1.3} is essentially Proposition 2.6.5 of \cite{CSS} in the non lax case.
\begin{proof}
Consider the expression $\iota_Q f^*\til\omega$. On the one hand, using chain map property of $f^*$ (or dg map property of $f$), we have
\begin{equation}\label{lemma 1.3 eq 3}
\iota_Q f^*\til\omega= f^*\iota_{\til Q}\til\omega=f^*\delta\til S=\delta f^* \til S  .
\end{equation}
On the other hand, we have
\begin{equation}\label{lemma 1.3 eq 4}
\iota_Q f^* \til\omega =\iota_Q (\omega+\LL_Q \beta)  .
\end{equation}
Using the fact that $\beta$ is a closed 2-form, we have $\iota_Q \LL_Q \beta = -\iota_Q \delta \iota_Q \beta$ which, compared with $\iota_Q \LL_Q \beta = \LL_Q \iota_Q \beta = \iota_Q\delta\iota_Q\beta-\delta \iota_Q\iota_Q \beta$, implies $\iota_Q \LL_Q \beta=-\frac12 \delta \iota_Q \iota_Q \beta$. Thus, continuing (\ref{lemma 1.3 eq 4}), we have
\begin{equation}
\iota_Q f^*\til \omega=\delta (S-\frac12 \iota_Q\iota_Q \beta) .
\end{equation}
Compared with (\ref{lemma 1.3 eq 3}), this implies that functions $f^*\til S$ and $S-\frac12 \iota_Q\iota_Q \beta$ on $\F$ have the same de Rham differential. Thus, they coincide up to a shift by a locally constant function. Since constants are concentrated in degree zero, we have equality (\ref{WE for S 1}) on the nose for $k\neq -1$ and up to a constant for $k=-1$. The proof of (\ref{WE for S 2}) is similar.
\end{proof}

\subsection{Relative version}
The following relative version of Definitions \ref{def: k-symp mfd}, \ref{def: WE manifolds} will be useful for us.
\begin{definition}\label{def: relative symp and Ham mfd}
Let  $(\F',Q',\omega')$ be a $(k+1)$-symplectic manifold and let $(\F,Q)$ be a dg manifold, $\pi\colon \F\ra \F'$ a dg map and $\omega$ a degree $k$ symplectic form on $\F$ which instead of (\ref{def of k-symp mfd: L_Q omega=0}) satisfies 
$$\LL_Q \omega=\pi^* \omega'.  $$
Then we say that $(\F,Q,\omega,\pi)$ is a $k$-symplectic manifold \emph{relative} to the $(k+1)$-symplectic manifold $(\F',Q',\omega')$.

If additionally one has $\omega'=\delta\alpha'$ (the symplectic form on $\F'$ is exact), $Q'$ has a Hamiltonian function $S'\in C^\infty(\F')_{k+2}$ and $\F$ is equipped with a function $S\in C^\infty(\F)_{k+1}$ satisfying
\begin{equation}
\iota_Q \omega=\delta S-\pi^* \alpha' ,
\end{equation}
then we say that $(\F,Q,\omega,S,\pi)$ is a $k$-Hamiltonian manifold relative to the exact $(k+1)$-Hamiltonian manifold $(\F',Q',\omega'=\delta\alpha',S')$.
\end{definition}

\begin{definition}
Let  $(\F,Q,\omega,\pi)$ and $(\til\F,\til Q,\til \omega,\til \pi)$ be two $k$-symplectic manifolds, both relative to the same $(k+1)$-symplectic manifold $(\F',Q',\omega')$. We define the weak equivalence between $(\F,Q,\omega,\pi)$ and $(\til\F,\til Q,\til \omega,\til \pi)$ as in the nonrelative case (Definition \ref{def: WE manifolds}), requiring additionally that $\til\pi\circ f=\pi$ and $\pi\circ g=\til\pi$.
\end{definition}

We remark that Lemma \ref{lemma 1.3} works in the relative case without any change.

\subsection{Linear algebra setting}
Here is a version Definitions \ref{def: k-symp mfd}, \ref{def: WE manifolds} adapted to the linear algebra case.
\begin{definition}
We call a \emph{degree $k$ Poincar\'e cochain complex} (or just \emph{$k$-Poincar\'e complex}) a cochain complex $\FF^\bullet$ over $\mathbb{R}$ with differential $d_Q$, equipped with a graded skew-symmetric non-degenerate pairing $\omega\colon \bigoplus_{i} \FF^i\otimes \FF^{-i-k} \ra \mathbb{R}$ satisfying
\begin{equation}\label{def of Poincare complex: omega(dx,y)=omega(x,dy)}
\omega(d_Q x,y)= (-1)^{|x|} \omega(x,d_Q y)
\end{equation}
for any homogeneous elements $x,y\in \FF$.
\end{definition}

A $k$-Poincar\'e complex can be seen as a $k$-symplectic  manifold with $Q=(d_Q)^*$ -- the pullback by $d_Q$ (in particular, $Q$ is a \emph{linear} cohomological vector field on $\F^\bt$) and with a constant $k$-symplectic form $\omega$. It also has a quadratic Hamiltonian function for $Q$, 
\begin{equation}
S(x)=\frac12 \omega(d_Q x,x) .
\end{equation}

\begin{definition}\label{def: WE linear}
A weak equivalence between two $k$-Poincar\'e complexes $(\FF,d_Q,\omega)$,  $(\til\FF,\til d_Q,\til \omega)$ is the following set of data:
\begin{enumerate}[(i)]
\item A pair of chain maps of complexes
\begin{equation}\label{f,g lin case}
f\colon \FF\ra \til\FF,\quad g\colon \til\FF\ra \FF 
\end{equation}
\item A pair of maps $H\colon \FF^\bt\ra \FF^{\bt-1}$, $\til H\colon \til\FF^\bt \ra \til\FF^{\bt-1}$ satisfying the chain homotopy property:
$$ d_Q H+H d_Q =\mr{id}-g\, f,\quad \til d_Q \til H + \til H \til d_Q=\mr{id}-f\, g $$
(thus, $f,g$ are quasi-isomorphisms inducing mutually inverse maps in cohomology).
\item A degree $k-1$ skew-symmetric bilinear form $\beta$ on $\FF$ and a degree $k-1$ skew-symmetric bilinear form $\til\beta$ on $\til\FF$
such that
\begin{eqnarray}
\label{f^* omega'=omega + L_Q alpha }
\til\omega(f(x),f(y))&=&\omega(x,y)+\\
\nonumber
&& +\underbrace{(-1)^{k} (\beta(d_Q x,y)+(-1)^{|x|+1} \beta(x,d_Q y))}_{(\LL_Q \beta)(x,y)},
\\
\label{g^* omega = omega' + L_Q alpha'}
\omega(g(\til x),g(\til y))&=&\til\omega(\til x,\til y)+\\
\nonumber
 &&+\underbrace{(-1)^{k} (\til\beta(\til d_Q \til x,\til y)+(-1)^{|\til x|+1} \til \beta(\til x,\til d_Q \til y))}_{(\LL_{\til Q}\til \beta)(\til x,\til y)}.
\end{eqnarray}
Here $x,y$ are test elements in $\FF$ and $\til x,\til y$ are in $\til\FF$. 
\end{enumerate}
\end{definition}



\begin{lemma}\label{lemma 1.5}
Let $(\FF,d_Q,\omega)$,  $(\til \FF,\til d_Q,\til \omega)$ be two $k$-Poincar\'e complexes.
Assume that we have a pair of quasi-inverse quasi-isomorphisms
${\FF \underset{g}{\stackrel{f}{\rightleftarrows}} \til\FF}$
with chain homotopies $H,\til{H}$ and assume additionally that $f^*\til\omega=\omega$. Then $(\FF,d_Q,\omega)$ and  $(\til\FF,\til d_Q,\til \omega)$ are weakly equivalent Poincar\'e complexes with $\beta=0$ and with $\til\beta$ defined in terms of the data $(\til\omega,\til d_Q,\til H)$ by 
\begin{multline}\label{alpha from Lemma 1.5}
\til \beta(\til x,\til y) = (-1)^{k+1}\Big(\left(\til \omega(\til H\til x,\til y)+(-1)^{|\til x|+1} \til\omega(\til x,\til H \til y)\right)+\\
+\frac12\left(-\til\omega(\til H \til x,\til H \til d_Q \til y) +(-1)^{|\til x|} \til \omega(\til H \til d_Q \til x, \til H\til y)\right)- \til \omega(\til H\til x,\til d_Q \til H\til y)\Big) .
\end{multline}
\end{lemma}
\begin{proof}
Property (\ref{f^* omega'=omega + L_Q alpha }) holds with $\beta=0$ by the assumption $f^*\til\omega=\omega$. Let us check \ref{g^* omega = omega' + L_Q alpha'}. We have
\begin{multline}\label{lemma 1.5 computation 1}
\omega(g(\til x),g(\til y))-\til\omega(\til x,\til y)\underset{\mr{using}\,f^*\til \omega=\omega}{=}\til \omega(fg(\til x),fg(\til y))-\til \omega(\til x,\til y)\\
=\omega\left((\mr{id}-\til d_Q \til H-\til H \til d_Q)(\til x),(\mr{id}-\til d_Q\til H-\til H \til d_Q)(\til y)\right)
-\til \omega(\til x,\til y)\\
=\underbrace{-\til \omega(\til d_Q \til H \til x,\til y)}_{a}\underbrace{-\til \omega(\til H \til d_Q\til x,\til y)}_{b}\underbrace{-\til\omega(\til x,\til d_Q \til H \til y)}_{c}\underbrace{-\til \omega(\til x,\til H \til d_Q \til y)}_{d}\\
+ 
\underbrace{\til \omega(\til d_Q \til H \til x, \til d_Q \til H \til y)}_{0}+\underbrace{\til \omega(\til d_Q \til H \til x, \til H \til d_Q \til y)}_{e}+ \underbrace{\til \omega(\til H\til d_Q \til x, \til H \til d_Q \til y)}_{f}+\underbrace{\til \omega(\til H\til d_Q \til x, \til d_Q \til H \til y)}_{g} .
\end{multline}
On the other hand, we have
\begin{multline}\label{lemma 1.5 computation 2}
(-1)^{k}(\til \beta(\til d_Q \til x,\til y)+(-1)^{|\til x|+1}\til \beta(\til x,\til d_Q \til y))\underset{(\ref{alpha from Lemma 1.5})}{=}
\\
=\underbrace{-\til \omega(\til H\til d_Q \til x,\til y)}_{b}+\underbrace{(-1)^{|\til x|+1}\til \omega(\til d_Q \til x,\til H\til y)}_{c}\\
+\underbrace{\frac12 \til \omega(\til H \til d_Q \til x,\til H \til d_Q \til y)}_{\frac12 f}
+\underbrace{(-1)^{|\til x|}\frac12\til \omega(\til H (\til d_Q)^2\til x, \til H\til y)}_0
+\underbrace{\til \omega(\til H\til d_Q \til x,\til d_Q \til H\til y)}_{g}\\
+\underbrace{(-1)^{|x'|}\omega'(H'x',d'_Qy')}_{a}\underbrace{-\til \omega(\til x,\til H\til d_Q \til y)}_{d}\\
+\underbrace{(-1)^{|\til x|+1}\frac12 \til \omega(\til H \til x, \til H (\til d_Q)^2 \til y)}_0 +\underbrace{\frac12 \til\omega(\til H\til d_Q \til x, \til H \til d_Q \til y)}_{\frac12 f}+\underbrace{(-1)^{|\til x|+1}(\til H\til x, \til d_Q \til H \til d_Q \til y)}_{e} .
\end{multline}
Comparing the terms in (\ref{lemma 1.5 computation 1}) and (\ref{lemma 1.5 computation 2}) and using $\til \omega(\til d_Q \til x,\til y)=(-1)^{|\til x|}\til \omega(\til x,\til d_Q \til y)$ we see that the expressions are equal and thus (\ref{g^* omega = omega' + L_Q alpha'}) holds.
\end{proof}

\subsection{Relative version in the linear case}
\begin{definition}
Let $(\F',d'_Q,\omega')$ be a $(k+1)$-Poincar\'e complex and let $(\F,d_Q)$ be a cochain complex equipped with a chain map $\pi\colon \F\ra\F'$ equipped with a degree $k$ constant symplectic structure $\omega\colon \wedge^2 \F\ra \RR$ satisfying 
$$ \omega(d_Q x,y)+ (-1)^{|x|+1}\omega(x,d_Q y)= (-1)^{k+1} \omega'(\pi(x),\pi(y))  $$
instead of (\ref{def of Poincare complex: omega(dx,y)=omega(x,dy)}). 
Then we say that $(\F,d_Q,\omega,\pi)$ is a $k$-Poincar\'e complex \emph{relative} to the $(k+1)$-Poincar\'e complex $(\F',d'_Q,\omega')$.
\end{definition}

\begin{definition}
Let $(\F,d_Q,\omega,\pi)$ and $(\til\F,\til{d}_Q,\til\omega,\til\pi)$ be two $k$-Poincar\'e complexes relative to the same $(k+1)$-Poincar\'e complex $(\F',d'_Q,\omega')$. We define a weak equivalence between $(\F,d_Q,\omega,\pi)$ and $(\til\F,\til{d}_Q,\til\omega,\til\pi)$ as in Definition \ref{def: WE linear} where additionally we require the properties 
$$\til\pi\circ f=\pi,\quad \pi\circ g=\til\pi,\quad \pi \circ H=0,\quad \til\pi \circ\til H=0.  $$
\end{definition}

We note that Lemma \ref{lemma 1.5} works in the relative case with no changes.

\section{Gluing via fiber products in a classical BV-BFV theory}
\subsection{Classical BV-BFV theories: a reminder}
\label{ss: class BV-BFV def}
In \cite{CMR1} the authors of the present paper and N. Reshetikhin introduced a refinement of Batalin--Vilkovisky formalism for gauge theories on manifolds with boundary, compatible with cutting--gluing, dubbed the ``BV-BFV formalism.'' For reader's convenience we recap the main definition.\footnote{We are following the version of the definition from \cite{MSW} in the non-extended case.
}

\begin{definition}\label{def: BV-BFV}
An $n$-dimensional classical BV-BFV theory $\mc{T}$ assigns to a closed $(n-1)$-manifold $\Sigma$ a ``BFV phase space'' -- a quadruple $(\Phi_\Sigma,Q_\Sigma,\omega_\Sigma=\delta\beta_\Sigma, S_\Sigma)$ consisting of: 
\begin{itemize}
\item A graded Fr\'echet 
manifold 
$\Phi_\Sigma$ of smooth sections of a graded vector bundle $E^\dd\ra \Sigma$.
\item 
A degree $+1$ local\footnote{
Here ``local'' means that the result of acting by $Q_\Sigma$ on a field at $x$ depends only on the jet of the field at $x$.
} cohomological vector field $Q_\Sigma$ on $\Phi_\Sigma$. 
\item An exact $0$-symplectic\footnote{In the context of Fr\'echet manifolds, we understand a symplectic structure as a  \emph{weakly} nondegenerate closed 2-form, where ``weakly nondegenerate'' means that the induced sharp map $\omega^\#\colon T\Phi\ra T^*\Phi$ is injective.} form $\omega_\Sigma=\delta\alpha_\Sigma$, where $\alpha_\Sigma=\int_\Sigma \underline\alpha_\Sigma$ with $\ul\alpha_\Sigma\in \Omega^{n-1,1}_\mr{loc}(\Sigma\times \Phi_\Sigma)$.\footnote{The subscript ``loc'' stands for locality: the evaluation of the form at $x\in \Sigma$ depends only on the jet of the field at $x$. The superscript $n-1,1$ is the de Rham bi-degree -- along $\Sigma$ and along $\Phi_\Sigma$. Equivalently, one can say that $\ul\alpha_\Sigma$ is a form on the jet bundle of $E^\dd$ of horizontal de Rham degree $n-1$ and vertical de Rham degree $1$.}
\item A degree $+1$ function $S_\Sigma=\int_\Sigma L_\Sigma$ -- a Hamiltonian for $Q_\Sigma$, with $L_\Sigma^\dd\in \Omega^{n-1,0}_\mr{loc}(\Sigma\times \Phi_\Sigma)$. 
\end{itemize}
In particular, $(\Phi_\Sigma,Q_\Sigma,\omega_\Sigma,S_\Sigma)$ is a $0$-Hamiltonian manifold.

To an $n$-manifold $M$ with boundary $\Sigma=\dd M$, the theory $\mc{T}$ assigns a quintuple $(\F_M,Q_M,\omega_M,S_M,\pi_{M,\Sigma})$ consisting of: 
\begin{itemize}
\item  A graded Fr\'echet 
 manifold $\F_M$ (``space of BV fields'') of smooth sections of a graded vector bundle $E\ra M$.
\item A degree $+1$ local cohomological vector field $Q_M$ on $\F_M$.
\item  A $(-1)$-symplectic form $\omega_M=\int_M\ul{\omega}_M$ on $\F_M$, with $\ul\omega_M\in \Omega^{n,2}_\mr{loc}(M\times \F_M)$. 
\item  A degree $0$ function $S_M=\int_M L_M$ on $\F_M$ (the BV action) of the form,  with $L_M\in \Omega_\mr{loc}^{n,0}(M\times \F_M)$.
\item  A surjective submersion $\pi_{M,\Sigma}\colon \F_M\ra \Phi_\Sigma$ satisfying the dg map property $Q_M \pi_{M,\Sigma}^* = \pi_{M,\Sigma}^* Q_\Sigma$.
\end{itemize}
Instead of being a Hamiltonian for $Q_M$, the function $S_M$ is assumed to satisfy the structure relation
\begin{equation}\label{BV-BFV eq}
\iota_{Q_M}\omega_M=\delta S_M-\pi_{M,\Sigma}^* \alpha_\Sigma ,
\end{equation}
linking the bulk data on $M$ and boundary data on $\Sigma$.
\end{definition}

In particular $(\F_M,Q_M,\omega_M,S_M,\pi_{M,\Sigma})$ is a $(-1)$-Hamiltonian manifold
relative to the exact $0$-Hamiltonian manifold $(\Phi_\Sigma,Q_\Sigma,\omega_\Sigma=\delta\alpha_\Sigma,S_\Sigma)$, in the sense of Definition \ref{def: relative symp and Ham mfd}.

\subsection{Gluing via fiber products}
Assume that we have an $n$-dimensional classical BV-BFV theory  $\mc{T}$ on a closed $n$-manifold $M$.\footnote{\label{footnote: M closed assumption}
The assumption that $M$ is closed is made for simplicity of exposition and is not essential. If $M$ has boundary, we need to replace $(-1)$-symplectic manifolds/Poincar\'e complexes below with their relative versions, relative to the boundary BFV data assigned by $\mc{T}$ to $\dd M$.
}
Let $M$ be cut by a closed $(n-1)$-submanifold $\Sigma$ into two $n$-manifolds with boundary, $M_1$ and $M_2$.
Let $\FF_M$ be the space of fields on $M$ and let 
\begin{multline} \label{fiber product}
\til{\FF}_{M,\Sigma}=\FF_{M_1}\times_{\Phi_\Sigma} \FF_{M_2}\\
 =
\{(\phi_1,\phi_2)\in \FF_{M_1}\times \FF_{M_2}\;\mr{s.t.}\; \pi_1(\phi_1)=\pi_2(\phi_2)\}
\end{multline}
be the fiber product (in the category of dg manifolds\footnote{Fiber products in the category of dg manifolds are not always well-defined, but in the case of (\ref{fiber product}) it is well-defined, since $\pi_{1,2}$ are surjective submersions. 
In the context of dg geometry, one usually gets to this nice case by replacing the dg manifolds appropriately with equivalent ones.
A related point: it would be natural to speak in terms of \emph{homotopy} fiber products in present context, but due to the fact that $\pi_{1,2}$ are surjective submersions, the strict fiber product provides a particular model for the homotopy fiber product.
}) of the spaces of fields on $M_1$ and $M_2$ over the phase space for the interface $\Sigma$. Here $\phi_{1,2}$ are fields on $M_1$ and $M_2$ and $\pi_{i}=\pi_{M_i,\Sigma}\colon \FF_{M_{i}}\ra \Phi_\Sigma$ with $i=1,2$ is the restriction of fields to the boundary (a structure map of the BV-BFV package). Put another way, $\til{\FF}_{M,\Sigma}$ fits into the fiber product diagram
$$
\begin{CD}
\til{\FF}_{M,\Sigma} @>p_2>> \FF_{M_2} \\
@Vp_1VV @V\pi_2VV \\
\FF_{M_1} @>\pi_1>> \Phi_\Sigma
\end{CD}
$$

One has a natural inclusion\footnote{Note that $i$ is not the identity: the right hand side of (\ref{i F to fiber product}) is strictly larger than the left hand side. Indeed, $\FF_M$ consists of smooth fields, whereas $\til{\FF}_{M,\Sigma}$ consists of fields smooth away from $\Sigma$ and with a special type of singularity allowed on $\Sigma$.}
\begin{equation}\label{i F to fiber product}
i\colon \FF_M\hra \til{\FF}_{M,\Sigma} .
\end{equation}
Let $\omega_M$ be the BV $(-1)$-symplectic structure on $\FF_M$ and let
\begin{equation}\label{til omega}
\til\omega_{M,\Sigma}=p_1^* \omega_{M_1}+p_2^* \omega_{M_2}
\end{equation} 
be the $(-1)$-symplectic structure on $\til{\FF}_{M,\Sigma}$ arising from $\omega_{M_1}$ and $\omega_{M_2}$ via fiber product construction.

\begin{conjecture}\label{conj: geometric}
The inclusion (\ref{i F to fiber product}) can be extended to a weak equivalence of $(-1)$-symplectic dg manifolds, in the sense of Definition \ref{def: WE manifolds}.
\end{conjecture}

A version adapted for a free (quadratic) BV-BFV theory is the following.
\begin{conjecture}\label{claim: lin case}
If $\mc{T}$ is a free BV-BFV theory, 
the inclusion (\ref{i F to fiber product}) can be extended to a weak equivalence of degree $-1$ Poincar\'e cochain complexes, in the sense of Definition \ref{def: WE linear}.
\end{conjecture}


The main result of this note is the proof of this conjecture in a collection of cases. More precisely, we prove the following.
\begin{thm}\label{thm main}
Let $\mathcal{T}$ be one of the following: 
\begin{itemize}
\item abelian Chern--Simons theory, 
\item abelian $BF$ theory,
\item $\p$-form electrodynamics in the first- or second-order formalism (including massless scalar field and usual electrodynamics  as $\p=0$ and $\p=1$ cases).
\end{itemize}
Then, given any open neighborhood $U\subset M$ of $\Sigma$, one can extend the inclusion (\ref{i F to fiber product}) 
to a package of maps $(i,p,H,\til{H})$:
\begin{equation}\label{ipHtilH package}
H\;\;\rotatebox[origin=c]{90}{$\curvearrowright$}\quad \FF_M \underset{p}{\stackrel{i}{\rightleftarrows}} \til{\FF}_{M,\Sigma}\quad \rotatebox[origin=c]{270}{$\curvearrowright$}\;\; \til{H} 
\end{equation}
such that:
\begin{enumerate}[(a)]
\item  \label{thm (a)} $i$ and $p$ are chain maps w.r.t. differentials $d_Q$ on $\FF_M$ and $\til{d_Q}$ on $\til\FF_{M,\Sigma}$:
\begin{eqnarray}
\til{d_Q}\, i &=& i\, d_Q, \label{d_Q i = i d_Q} \\
d_Q\, p &=& p\, \til{d_Q}. \label{d_Q p = p d_Q}
\end{eqnarray}
\item \label{thm (b)} $H$ is the chain homotopy between identity and $p\,i$ on $\FF_M$ and $\til{H}$ is the chain homotopy between identity and $i\, p$ on $\til\FF_{M,\Sigma}$:
\begin{eqnarray}
d_QH+Hd_Q &=& \mr{id}-p\, i, \label{chain homotopy rel H} \\
\til{d_Q}\til{H}+\til{H} \til{d_Q} &=&\mr{id}-i\,p . \label{chain homotopy rel tilH}
\end{eqnarray}
\item \label{thm (c)} The package of maps $(i,p,H,\til{H})$
is \emph{local near $\Sigma$}: $p$ is smoothing in the neighborhood  $U$ of $\Sigma$ and is identity outside $U$.\footnote{Put another way, $p$ is an integral operator with distributional kernel which is smooth in $U\times U$ and the Dirac delta-form on the diagonal outside $U\times U$.} The operators $H,\til{H}$ vanish on fields supported outside $U$.
\item The $(-1)$-symplectic structures on $\FF_M$ and $\til\FF_{M,\Sigma}$ are related by 
\begin{eqnarray}
i^* \til\omega &=& \omega, \label{i^* til omega = omega} \\
p^* \omega &=& \til\omega + \LL_{\til{Q}} \til\beta, \label{p^* omega = til omega+ L_Q alpha}
\end{eqnarray}
with $\til\beta$ given by (\ref{alpha from Lemma 1.5}).
\end{enumerate}
\end{thm}

In particular, relations (\ref{d_Q i = i d_Q})--(\ref{chain homotopy rel tilH}) imply that $i$ and $p$ are quasi-isomorphisms (mutually inverse in cohomology). 

The proof of this theorem (and the construction of the corresponding maps) is given in the subsequent sections: 
the case of abelian Chern--Simons and abelian $BF$ is covered in Section \ref{s: ex - abCS}. The case of $\p$-form theory in the first-order formalism---in Section \ref{ss: p-form 1st order local homotopy} and the second-order formalism---in Section \ref{ss: p-form 2nd order} (``local-near-$\Sigma$ version'').

We note that (\ref{i^* til omega = omega}) holds trivially by construction (\ref{til omega}), while (\ref{p^* omega = til omega+ L_Q alpha}) follows from items (\ref{thm (a)}), (\ref{thm (b)}), (\ref{thm (c)}) and Lemma \ref{lemma 1.5}.

\begin{remark}
In Theorem \ref{thm main} we \emph{do not} require the property $p\, i=\mr{id}$ (in fact, it doesn't hold in the examples discussed in this paper); similarly, we do not require that $i\, p$ is an idempotent.
\end{remark}

\begin{remark}
We want to emphasize the significance of item (\ref{thm (c)}) in  Theorem \ref{thm main}: it implies functoriality of the weak equivalence with respect to Segal's functorial picture of QFT. Put another way, the  data of the weak equivalence is localized near the cut $\Sigma$ and does not interact with cutting--gluing away from $\Sigma$.

A related point is that although in the beginning of this section we asked 
$M$ to be closed 
(for simplicity and so that we can use the convenient language of Poincar\'e complexes), 
we can without problem allow $M$ to have boundary, disjoint from the hypersurface $\Sigma$, cf. footnote \ref{footnote: M closed assumption}.
\end{remark}

In a future paper we plan to study examples of Conjecture \ref{conj: geometric} coming from interacting field theories (in the spirit of Theorem \ref{thm main}, with complexes of fields 
enhanced to
$L_\infty$ algebras of fields).

\subsection{Sketch of proof of Conjecture \ref{conj: geometric} by factoring through cohomology (and ignoring locality near $\Sigma$)}
\label{ss: sketch of proof of Conj 2.2}
The following construction was independently suggested to us by Domenico Fiorenza and Bruno Vallette.

Consider the setup of Conjecture \ref{conj: geometric}. We will assume that the BV-BFV theory $\mc{T}$ is a perturbation of a free theory $\mc{T}_0$ in which Conjecture \ref{claim: lin case} holds.\footnote{We understand that $\mc{T}$ is obtained from $\mc{T}_0$ by adding terms of polynomial degree $\geq 3$ in fields to the quadratic action of $\mc{T}_0$ and, respectively, perturbing the linear cohomological vector field $Q_0=d_Q$ of $\mc{T}_0$ by terms of polynomial degree $\geq 2$ in fields. For simplicity we assume that the $(-1)$-symplectic form does not get deformed.} In particular, $\F_M$ is a graded vector space with cohomological vector field $Q_M$ vanishing at the origin and corresponding to an $L_\infty$ algebra structure on $\F_M[-1]$;\footnote{For the sake of convenience, below in this section we will  omit the degree shift $[-1]$ in notations, and speak of $\F$ as an $L_\infty$ algebra. We will also use the terms ``dg map'' and ``$L_\infty$ map'' interchangeably.} $\omega_M$---a constant $(-1)$-symplectic form---corresponds to a degree $-3$ inner product on $\F_M[-1]$ with respect to which the $L_\infty$ operations have cyclic property. 

We have a strict (i.e. having only linear component) $L_\infty$ map (\ref{i F to fiber product}) and we want to promote it to a weak equivalence in the sense of Definition \ref{def: WE manifolds}.

We employ the construction of \cite[Section 10.4.6]{LV} adapted to our setting: 
we consider the cohomology\footnote{
We use the font to distinguish cohomology $\HH$ from chain homotopies $H$.
}
$\HH(\F)$  of $\F$  with respect to $d_Q$ (the linear part of the $Q$) and the cohomology $\HH(\til\F)$ of $\til\F$ with respect to the linear part $\til{d_Q}$ of $\til Q$.
We choose contractions 
\begin{equation}\label{contractions}
{h\;\rotatebox[origin=c]{90}{$\curvearrowright$}\; (\F,d_Q) \underset{j}{\stackrel{r}{\rightleftarrows}} \HH(\F)},\qquad {\til{h}\;\rotatebox[origin=c]{90}{$\curvearrowright$}\; (\til\F,\til{d_Q}) \underset{\til{j}}{\stackrel{\til{r}}{\rightleftarrows}} \HH(\til\F)} 
\end{equation} 
such that 
\begin{equation}\label{j,r commute with i}
i\circ j= \til{j}\circ i_*,\quad i_* \circ r =\til{r}\circ i,\quad i\circ h=\til{h}\circ i,
\end{equation}
where 
\begin{equation}\label{i_*}
i_*\colon \HH(\F)\ra \HH(\til\F)
\end{equation} 
is the linear isomorphism in cohomology induced from $i$ ($i_*$ is an isomorphism by assumption that the free theory $\mc{T}_0$ satisfies Conjecture \ref{claim: lin case}).
The homological perturbation lemma\footnote{See \cite[Section 6.4]{KS}, \cite{Markl} and \cite[Section 10.3]{LV} on the homotopy transfer of $\infty$-algebras.} applied at the level of symmetric coalgebras deforms the maps (\ref{contractions}) to
\begin{equation}\label{contractions deformed}
\mc{H}\;\rotatebox[origin=c]{90}{$\curvearrowright$}\; (\F,Q) \underset{\mc{J}}{\stackrel{\mc{R}}{\rightleftarrows}} (\HH(\F),Q_\HH),\qquad \til{\mc{H}}\;\rotatebox[origin=c]{90}{$\curvearrowright$}\; (\til\F,\til{Q}) \underset{\til{\mc{J}}}{\stackrel{\til{\mc{R}}}{\rightleftarrows}} (\HH(\til\F),\til{Q_{\HH}}) .
\end{equation} 
Here $\mc{J,R,\til{J},\til{R}}$ are $L_\infty$ maps, $\mc{H,\til{H}}$ -- $L_\infty$ homotopies and $Q_\HH,\til{Q_\HH}$ -- induced minimal $L_\infty$ algebra structures on cohomology.


By the fiber product construction (\ref{fiber product}), $i$ intertwines the $L_\infty$ structures $i\circ Q=\til{Q}\circ i\colon \mr{Sym}_\mr{co}\F\ra \mr{Sym}_\mr{co}\til\F$. This implies that the map (\ref{i_*}) is a strict $L_\infty$ isomorphism of minimal $L_\infty$ algebras $\HH(\F)$, $\HH(\til\F)$. Also, relations (\ref{j,r commute with i}) become 
\begin{equation}\label{J,R commute with i}
i \circ \mc{J} = \til{\mc{J}}\circ i_*,\quad i_*\circ \mc{R}=\til{\mc{R}}\circ i,\quad i\circ\mc{H}=\til{\mc{H}}\circ i.
\end{equation} 

%

Next---this is the key step borrowed from \cite[Section 10.4.6]{LV}---we construct the $L_\infty$ map
\begin{equation}\label{inverse q-iso (Loday-Vallette)}
p=\mc{J}\circ (i_*)^{-1} \circ \mc{\til{R}}\colon \quad \til{\F}\ra \F .
\end{equation}
It is an $L_\infty$ quasi-inverse of the canonical inclusion (\ref{i F to fiber product}). More precisely, one has the diagram of maps 
\begin{equation}\label{ipHtilH package for L_infty algebras}
\mc{H}\;\;\rotatebox[origin=c]{90}{$\curvearrowright$}\quad (\FF,Q) \underset{p}{\stackrel{i}{\rightleftarrows}} (\til{\FF},\til{Q})\quad \rotatebox[origin=c]{270}{$\curvearrowright$}\;\; \til{\mc{H}} ,
\end{equation}
where $i$ and $p$ are mutually quasi-inverse $L_\infty$ quasi-isomorphisms and $\mc{H}$, $\til{\mc{H}}$ are the respective $L_\infty$ homotopies, i.e., one has 
\begin{eqnarray}
\mr{id}-p\circ i&=&Q \mc{H}+\mc{H} Q , \\
\mr{id}-i\circ p&=& \til{Q} \til{\mc{H}}+\til{\mc{H}} \til{Q} . \label{id-i p = QH+HQ}
\end{eqnarray}
These relations follow immediately from the construction (\ref{inverse q-iso (Loday-Vallette)}) and from (\ref{J,R commute with i}).

The package of maps (\ref{ipHtilH package for L_infty algebras}), upon dualization and passing to the completion of symmetric algebras to algebras of smooth functions, gives rise to 
a chain equivalence of chain complexes
\begin{equation}\label{ipHtilH package for C^infty}
\mc{H}\;\;\rotatebox[origin=c]{90}{$\curvearrowright$}\quad (C^\infty(\FF),Q) \underset{p^*}{\stackrel{i^*}{\leftrightarrows}} (C^\infty(\til{\FF}),\til{Q})\quad \rotatebox[origin=c]{270}{$\curvearrowright$}\;\; \til{\mc{H}} . 
\end{equation}
Thus, we have (\ref{def: WE manifolds (a)}) of Definition \ref{def: WE manifolds} of weak equivalence between $\F$ and $\til\F$.

For (\ref{def: WE manifolds (b)}) of Definition \ref{def: WE manifolds}, we note that, by a straightforward extension from functions to differential forms\footnote{\label{footnote 11}
The tangent lift is constructed as follows: we have ``doubled'' contractions (\ref{contractions}), ${h\oplus h\;\rotatebox[origin=c]{90}{$\curvearrowright$}\; (T[1]\F,d_Q\oplus d_Q) \underset{j\oplus j}{\stackrel{r\oplus r}{\rightleftarrows}} T[1]\HH(\F)}$ and similar one for tilde-side. We think of $T[1]\F$ as $\F\oplus \F[1]$ with an extra differential $\delta$ mapping between the two copies of $\F$. Then we pass to (completed) symmetric algebras of the duals and turn on the perturbation of the differential $L_{d_Q}\ra L_Q$ via homological perturbation lemma. This yields the contraction 
$\mc{H}^\mr{lifted}\;\rotatebox[origin=c]{90}{$\curvearrowright$}\; (\Omega^\bt(\F),\LL_Q) \underset{\mc{J}^*}{\stackrel{\mc{R}^*}{\leftrightarrows}} (\Omega^\bt(\HH(\F)),\LL_{Q_\HH})$ and similar for tilde-side. Then we repeat the construction (\ref{inverse q-iso (Loday-Vallette)}) at the level of $T[1]\F$, $T[1]\til\F$. This gives us the maps of the chain equivalence (\ref{chain equivalence of forms}). 
We remark that all maps in (\ref{chain equivalence of forms}) commute with $\delta$ by construction, since the maps of the non-perturbed ``doubled'' contractions above commute with $\delta$ and since the differential perturbation commutes with $\delta$ also.
%
%
} (a.k.a. tangent lift from $\F,\til\F$ to the shifted tangent bundles $T[1]\F$, $T[1]\til{F}$), one gets a chain equivalence
\begin{equation}\label{chain equivalence of forms}
\mc{H}^\mr{lifted}\;\;\rotatebox[origin=c]{90}{$\curvearrowright$}\quad (\Omega^\bt(\FF),\LL_Q) \underset{p^*}{\stackrel{i^*}{\leftrightarrows}} (\Omega^\bt(\til{\FF}),\LL_{\til{Q}})\quad \rotatebox[origin=c]{270}{$\curvearrowright$}\;\; \til{\mc{H}}^\mr{lifted} .
\end{equation}
We know that, by fiber product construction, we have $i^*\til\omega=\omega$. We also have
\begin{equation}
p^*\omega=p^*i^*\til\omega = (\mr{id}-\LL_{\til{Q}} \til{\mc{H}}^\mr{lifted}-\til{\mc{H}}^\mr{lifted}\LL_{\til{Q}})\til\omega=\til\omega + \LL_{\til{Q}} \til\beta
\end{equation}
with $\til\beta=-\til{\mc{H}}^\mr{lifted}\til\omega$. Here we used $\LL_{\til{Q}} \til\omega=0$ and (the tangent lift of) the equation (\ref{id-i p = QH+HQ}). Note that $\til\beta$ is a $\delta$-closed $2$-form, since $\delta\til\omega=0$ and since $\til{\mc{H}}^\mr{lifted}$ commutes with $\delta$ (cf. footnote \ref{footnote 11}). Thus, (\ref{def: WE manifolds (b)}) of Definition \ref{def: WE manifolds} holds with $\beta=0$ and $\til\beta=-\til{\mc{H}}^\mr{lifted}\til\omega$.

\begin{remark}
The construction given in this section is not compatible with locality near the hypersurface $\Sigma$ in the sense of (\ref{thm (c)}) of Theorem \ref{thm main}. We hope to find a proof of Conjecture \ref{conj: geometric} by presenting a weak equivalence that is local near $\Sigma$.
\end{remark}

\section{``Smearing'' homotopy in de Rham theory}
The following lemma is rephrased from Kontsevich-Soibelman, \cite[section 6.5, lemma 2]{KS}.\footnote{
See also \cite{HL}.}
\begin{lemma}\label{lemma (smearing)}
Let $N$ be an oriented (not necessarily compact) $n$-manifold and $V\subset N\times N$ 
a tubular neighborhood of the diagonal $\mr{Diag}\subset N\times N$ with projection $r\colon V\ra \mr{Diag}$. Then:
\begin{enumerate}
\item \label{lm (1)}
There exists a smooth closed $n$-form $\rho$ on $N\times N$ satisfying
\begin{itemize}
\item $\mr{supp}(\rho)\subset V$,
\item 
$r_* \rho=1$ (the integral of $\rho$ over each fiber of $r$ is $1$).\footnote{If $N$ is compact, an equivalent statement is that the cohomology class $[\rho]\in H^n_c(V)$ is Poincar\'e dual to the homology class of the diagonal $[\mr{Diag}]\in H_n(V)$.}
\end{itemize}
\item Denote $R\colon \Omega_\mr{distr}^\bt(N)\ra \Omega^\bt(N)$ the integral operator from distributional to smooth forms determined by the integral kernel $\rho$. Then
\begin{itemize}
\item $R$ is a chain map w.r.t.  the de Rham differential.
\item 
Let $\iota\colon \Omega^\bt(N)\hra \Omega^\bt_\mr{distr}(N)$ be the canonical inclusion of smooth forms into distributional forms. Then $\iota R$ and $R\iota$ are chain-homotopic to identity. More precisely,
there exists a distributional form $\chi\in \Omega^{n-1}_\mr{dist}(N\times N)$ supported in $V$ and smooth away from $\mr{Diag}$ such that $d\chi=\delta_\mr{Diag}-\rho$ (as distributional forms). As a consequence, one has
\begin{eqnarray}
\mr{id}_{\Omega_\mr{distr}^\bt}- \iota R &=& d\kappa+\kappa d , \label{lemma KS chain homotopy rel 1} \\
\mr{id}_{\Omega^\bt}- R \iota  &=& d\kappa|_{\Omega_\bt}+\kappa|_{\Omega_\bt} d . \label{lemma KS chain homotopy rel 2}
\end{eqnarray}
Here 
$\kappa\colon \Omega^\bt_\mr{distr}(N)\ra \Omega^{\bt-1}_\mr{distr}(N)$ is the operator defined by the integral kernel $\chi$.

\end{itemize}

\end{enumerate}
\end{lemma}

We call the operator $R$ a ``smearing'' (or ``smoothing,'' or ``mollifying'') operator.

\begin{proof}
A form $\rho$ satisfying conditions (\ref{lm (1)}) can be constructed e.g. by taking a Mathai-Quillen representative  $\tau\in \Omega^n(\N\mr{Diag})$ of the Thom class of the normal bundle of $\mr{Diag}$ (see \cite{MQ}) and choosing an embedding $\phi\colon \N\mr{Diag}\hra V$.\footnote{We denote the normal bundle by $\N$ to avoid confusion with the manifold $N$.} 
Then one can define 
$$
\rho = \left\{
\begin{array}{ccc}
(\phi^{-1})^* \tau & \mr{in} & \phi(\N\mr{Diag}), \\
0  & \mr{outside} & \phi(\N\mr{Diag})
\end{array}
\right.
$$
Likewise, one can construct the distributional form  $\chi$ as follows: first construct $\bar\chi\in \Omega^{n-1}_\mr{distr}(\N\mr{Diag})$ as 
$\bar\chi= \pi_* s^* \tau$ where
$$
\begin{CD}
[1,+\infty) \times \N\mr{Diag} @>{s}>> \N\mr{Diag}\\
@V{\pi}VV \\ 
\N\mr{Diag} 
\end{CD}
 $$
where $s$ stretches the vector in the fiber of $\N\mr{Diag}\ra \mr{Diag}$ by a factor in $[1,+\infty)$; $\pi$ is the projection to the second factor; by Stokes' theorem one has $d\bar\chi=\delta_{\mr{zero\;section}}-\tau$. Then one sets $\chi=\phi_*\bar\chi$.
\end{proof}

\section{Abelian Chern--Simons theory}
\label{s: ex - abCS}
Abelian Chern--Simons theory in BV-BFV formalism (we refer to \cite{CMR1} for details) assigns to 
an oriented compact 3-manifold $M$ 
the space of fields 
$$\F_M=\Omega^\bt(M)[1]\quad\ni \mc{A}$$
-- the degree shifted de Rham complex,
equipped with the quadratic BV action
$$S_M=\frac12 \int_M \mc{A}\wedge d\mc{A},$$
the linear cohomological vector field $Q_M=\int_M d\mc{A}\frac{\delta}{\delta\mc{A}}$ and the constant $(-1)$-symplectic form $\omega_M=-\frac12 \int_M \delta\mc{A}\wedge \delta\mc{A}$.

To a surface $\Sigma$ cutting $M$ into $M_1$ and $M_2$, the theory assigns the phase space
$$ \Phi_\Sigma = \Omega^\bt(\Sigma)[1]\quad \ni \mc{A}_\Sigma $$ 
equipped
with the symplectic form $\omega_\Sigma=\delta\left(\frac12 \int_\Sigma \mc{A}_\Sigma\wedge \delta\mc{A}_\Sigma\right)$, the cohomological vector field $Q_\Sigma=\int_\Sigma d\mc{A}_\Sigma\frac{\delta}{\delta \mc{A}_\Sigma}$ and the Hamiltonian $S_\Sigma={\frac12 \int_\Sigma \mc{A}_\Sigma \wedge d\mc{A}_\Sigma}$. Projections $\pi_{1,2}\colon \F_{M_{1,2}}\ra \Phi_\Sigma$ are given by pulling back a form on $M_{1,2}$ to $\Sigma$.


The canonical inclusion (\ref{i F to fiber product}) is:
\begin{equation*}
\FF_M=\Omega^\bt(M)[1] \qquad \stackrel{i}{\hra} \qquad \til{\FF}_{M,\Sigma}= \Omega^\bt(M_1)[1]\times_{\Omega^\bt(\Sigma)[1]} \Omega^\bt(M_2)[1],
\end{equation*}
where left and right side are complexes w.r.t. the de Rham differential on $M$.

We construct a chain map $p\colon \til{\FF}_{M,\Sigma}\ra \FF_M$ as follows. Fix an embedding $\psi\colon \N\Sigma\ra M$ of the normal  bundle of $\Sigma$ into $M$, with image $U=\psi(\N\Sigma)$ an open neighborhood of $\Sigma$ (which can be made as thin as needed). We apply Lemma \ref{lemma (smearing)} to $N=\N\Sigma$ and consider the corresponding forms $\rho,\chi$. Define a distributional $n$-form $\wh\rho$ on $M\times M$ by
 $$ \wh\rho=\left\{  
\begin{array}{ccc}
(\psi^{-1}\times \psi^{-1})^*
\rho 
 & \mr{in} & U\times U ,\\
\delta_{\mr{Diag}_M} & \mr{outside} & U\times U 
\end{array}
 \right.
 $$
and define a distributional $(n-1)$-form $\wh\chi$ on $M\times M$   as
 $$ \wh\chi=\left\{  
\begin{array}{ccc}
(\psi\times \psi)_*
\chi 
 & \mr{in} & U\times U ,\\
0 & \mr{outside} & U\times U 
\end{array}
 \right.
 $$

We define the map $p\colon \til{\FF}_{M,\Sigma}\ra \FF_M$ as the integral operator with kernel $\wh\rho$ (thus, $p$ is smoothing inside the neighborhood $U$ of $\Sigma$ and identity outside); since $\wh\rho$ is a closed distributional form on $M\times M$, $p$ is indeed a chain map. Likewise, we define the chain homotopy $\til{H}\colon \til{\FF}_{M,\Sigma}\ra \til{\FF}_{M,\Sigma}[-1]$ as the integral operator with kernel $\wh{\chi}$. Denote its restriction to $\FF_M$ by $H$. The chain homotopy identities (\ref{chain homotopy rel H}), (\ref{chain homotopy rel tilH}) are immediate consequences of (\ref{lemma KS chain homotopy rel 1}), (\ref{lemma KS chain homotopy rel 2}).

Thus we have the quadruple of maps (\ref{ipHtilH package}) satisfying all points of the Theorem \ref{thm main}.

%
%

\subsection{Abelian $BF$ theory}
The case of abelian $BF$ theory is a trivial modification of the construction for abelian Chern--Simons. Now $M$ can be of any dimension $n$; the BV action is 
$$S_M=\int_M \mc{B}\wedge d\mc{A}$$ 
and 
the space of fields consists of two copies  of de Rham complex: 
$$\FF_M=\Omega^\bt(M)[\p]\oplus \Omega^\bt(M)[n-\p-1]\quad\ni (\mc{A},\mc{B})$$ 
(for some fixed shift $\p\in\mathbb{Z}$) with differential $d_Q=d\oplus d$ and with constant symplectic form $\omega=(-1)^n\int_M \delta\mc{B}\wedge \delta\mc{A}$. 

The BFV phase space is 
$$\Phi_\Sigma=\Omega^\bt(\Sigma)[\p]\oplus \Omega^\bt(\Sigma)[n-\p-1]$$ 
and the projections $\pi_{1,2}\colon \FF_{M_{1,2}}\ra \Phi_\Sigma$ are given by the pullback of $\mc{A}$- and $\mc{B}$-form from $M_{1,2}$ to $\Sigma$.
The canonical inclusion (\ref{i F to fiber product}) is
\begin{equation*}
\begin{CD}
\FF_M=\Omega^\bt(M)[\p]\oplus \Omega^\bt(M)[n-\p-1] \\
@V{i}VV \\
\til\FF_{M,\Sigma} =
\left(\Omega^\bt(M_1)\times_{\Omega^\bt(\Sigma)} \Omega^\bt(M_2) \right)[\p] \oplus \left(\Omega^\bt(M_1)\times_{\Omega^\bt(\Sigma)} \Omega^\bt(M_2) \right)[n-\p-1] 
\end{CD}
\end{equation*}

We proceed as in abelian Chern--Simons and set $p\colon \til\FF_{M,\Sigma}\ra \FF_M$ to be given by $p$ from above (the integral operator with kernel $\wh\rho$) acting on both $\mc{A}$- and $\mc{B}$-forms in $\til\FF_{M,\Sigma}$.
Likewise, we define the homotopies to be given by $\til{H}$ and $H$ from above (defined by the integral kernel $\wh\chi$) acting diagonally on both components in $\til\FF_{M,\Sigma}$ or $\FF_M$, respectively.

\section{Massless scalar field in the second-order formalism}
\label{s: scalar field}
As a warm-up for the $\p$-form electrodynamics for general $\p$ (coming in Section \ref{s: p-form ED} below), let us first discuss the case $\p=0$, i.e., massless scalar field.

Massless scalar field on a Riemannian $n$-manifold $M$ is defined by the BV action $S_M=\int_M \frac12 d\phi\wedge *d\phi$. The space of bulk fields is the Poincar\'e complex\footnote{See footnote \ref{footnote: d*d sign in second order p-form theory} below for the explanation of the reason why we write the differential as $d*^{-1}d$ rather than $d*d$, and why the projection to boundary field $\pi$ contains $*^{-1}d$ rather than $*d$.}
\begin{equation}\label{scalar field F}
\FF_M=\qquad \Big(\;\; \underset{\phi}{\Omega^0(M)}\xra{d*^{-1}d} \underset{\phi^+}{\Omega^{n}(M) } \;\; \Big)
\end{equation}
with constant $(-1)$-symplectic form $\omega_M=(-1)^n\int_M \delta \phi^+\wedge\delta \phi$. 
The boundary phase space for a closed $(n-1)$-submanifold $\Sigma\subset M$ (splitting $M$ into $M_1$ and $M_2$) is 
$$\Phi_\Sigma = \underset{\phi_\Sigma}{\Omega^0(\Sigma)}\oplus \underset{P_\Sigma}{\Omega^{n-1}(\Sigma)}$$
with zero differential and $0$-symplectic form $\omega_\Sigma=-\int_\Sigma \delta\phi_\Sigma\wedge \delta P_\Sigma$. The projection $\pi\colon \FF_{M_{1,2}}\ra \Phi_\Sigma$ maps $(\phi,\phi^+)$ to $(\phi|_\Sigma,(*^{-1}d\phi)|_\Sigma)$.

The fiber product (\ref{fiber product}) is 
\begin{equation}\label{scalar field fiber product}
\til\FF_{M,\Sigma}=\qquad \Big(\;\;\Omega^0(M_1)\underset{\Omega^0(\Sigma)\oplus \Omega^{n-1}(\Sigma)}{\times} \Omega^0(M_2) \;\; \xra{d*^{-1}d} \;\; \Omega^n(M_1)\times \Omega^n(M_2) \;\;\Big) .
\end{equation}
The first term of this complex is the space of functions on $M$, smooth on $M_1$ and $M_2$, continuous and differentiable across $\Sigma$. The second term consists of $n$-forms on $M$ smooth on $M_1$ and $M_2$ and no regularity condition at $\Sigma$.

We have, as always, the canonical inclusion $i\colon \FF_M\hra \til\FF_{M,\Sigma}$ of fields on $M$ into the fiber product. 

\subsection{Construction 1 (Hodge-theoretic)} \label{ss: scalar field Hodge homotopy}
For the map 
$p\colon \til\FF_{M,\Sigma}\ra \FF_M$ we take 
$$p=e^{-\epsilon\Delta}$$ -- the heat flow in time $\epsilon>0$, acting on $0$- and $n$-forms from the fiber product (\ref{scalar field fiber product}). Here $\Delta=dd^*+d^*d$ is the Hodge Laplacian acting on forms on $M$ (in the non-negative convention). We understand $\epsilon$ as a regulator or smearing parameter.

\begin{remark} Unlike the example of Section \ref{s: ex - abCS}, here $p$ is not identity away from a neighborhood of $\Sigma$ -- it is a smearing operator everywhere on $M$. 
Note that we cannot naively use $p$ of Section \ref{s: ex - abCS} since it is not 
a chain map w.r.t. the differential $d*d$ in the chain complexes (\ref{scalar field F},\ref{scalar field fiber product}).
\end{remark}

We construct the chain homotopies $H\;\;\rotatebox[origin=c]{90}{$\curvearrowright$} \;\;\FF_M$ and $\til{H}\;\;\rotatebox[origin=c]{90}{$\curvearrowright$}\;\; \til\FF_{M,\Sigma}$ as follows:
\begin{equation}\label{scalar field H}
H\colon\qquad \Omega^0(M)\xleftarrow{(-1)^{n}(\mr{id}-e^{-\epsilon \Delta})\Delta^{-1}*} \Omega^n(M) 
\end{equation}
and 
\begin{equation}\label{scalar field tilde H}
\til{H}\colon\qquad \til\Omega^0(M,\Sigma)\xleftarrow{(-1)^{n}(\mr{id}-e^{-\epsilon \Delta})\Delta^{-1}*} \til\Omega^n(M,\Sigma)  .
\end{equation}
Here we denoted the two terms of the fiber product complex (\ref{scalar field fiber product}) by $\til\Omega^{0,n}(M,\Sigma)$. We understand the operator $\Delta^{-1}$ as zero on harmonic forms and as the inverse of $\Delta$ on the orthogonal complement of harmonic forms (where $\Delta$ is strictly positive). 

To summarize the result, the package of maps $(i,p,H,\til{H})$ we just constructed satisfies (\ref{thm (a)}) and (\ref{thm (b)}) of Theorem \ref{thm main}, but does not satisfy (\ref{thm (c)}).

\subsection{Construction 2 (local near $\Sigma$)}
\label{ss: scalar field local homotopy}

\begin{lemma}\label{lemma 4.2}
 Let $M$ be a Riemannian $n$-manifold cut into $M_1$ and $M_2$ by a closed $(n-1)$-submanifold $\Sigma\subset M$. Let $U\subset M$ be an open neighborhood of $\Sigma$ in $M$. Then there exists a distributional $0$-form $\chi\in \Omega^0_\mr{distr}(M\times M)$ such that
\begin{enumerate}[(a)]
\item \label{lemma 4.2 (a)} 
$\chi$ vanishes outside $U\times U$.
\item \label{lemma 4.2 (b)} $\chi$ is smooth on $(M\times M)\backslash \mr{Diag}_{\bar{U}}$.
\item \label{lemma 4.2 (c)} $\chi$ satisfies  
$$(d*^{-1}d)_1\chi=\left\{ 
\begin{array}{cl}
\delta_\mr{Diag}^{(n,0)}+\mr{smooth}\; (n,0)\mr{-form} & \mr{in}\; U\times U, \\
0 & \mr{outside}
\end{array}
 \right.$$
 and
 $$(d*^{-1}d)_2\chi=\left\{ 
\begin{array}{cl}
\delta_\mr{Diag}^{(0,n)}+\mr{smooth}\; (0,n)\mr{-form} & \mr{in}\; U\times U, \\
0 & \mr{outside}
\end{array}
 \right.$$
where where the subscript in $(d*^{-1}d)_{1,2}$ means that the operator acts on the first or second argument of $\chi$ (first or second copy of $M$ in $M\times M$). Bi-index $(i,j)$ refers to the de Rham bi-degree of a form on $M\times M$, i.e., $\Omega^{i,j}(M\times M)=\Omega^i(M)\widehat{\otimes}\Omega^j(M)$; $\delta_\mr{Diag}^{(i,j)}$ refers to the $(i,j)$-component of the delta-form on the diagonal in $M\times M$.
\end{enumerate}
\end{lemma}


\begin{proof}
Choose a smooth function  $\mu$ on $M\times M \backslash \mr{Diag}_{\dd U}$ (a ``bump function'') such that
\begin{enumerate}
\item $\mu=0$  outside $U\times U$,
\item $\mu=1$ in some open neighborhood $V$ of $\mr{Diag}_U$ contained in $U\times U$.
\end{enumerate}
Such a $\mu$ can be easily constructed using a partition of unity: 
take some covering of $M\times M\backslash \mr{Diag}_{\dd U}$ by open sets $\{v_\alpha\}$ such that if $v_\alpha$ intersects $\mr{Diag}_U$ then it is contained in $U\times U$ and let $\{\psi_\alpha\}$ be a partition of unity subordinate to this cover. Then we can define 
$$\mu=\sum_{\alpha\;\mr{s.t.}\,v_\alpha\cap \mr{Diag}_U\neq \varnothing} \psi_\alpha .$$

Let $G\in C^0(U\times U)$ be the Green function for the Laplacian $\Delta$ on $U$, with Dirichlet boundary condition on $\dd U$.
Then we define 
$$\chi :=\left\{
\begin{array}{ccc}
(-1)^n\mu\cdot G, & \mr{in} & U\times U, \\
0 & \mr{outside} & U\times U 
\end{array}
\right.$$
With this definition, properties (\ref{lemma 4.2 (a)}), (\ref{lemma 4.2 (b)}), (\ref{lemma 4.2 (c)}) are obvious.
\end{proof}

We define the maps $$ H\;\;\rotatebox[origin=c]{90}{$\curvearrowright$}\quad \FF_M \underset{p}{\stackrel{i}{\rightleftarrows}} \til{\FF}_{M,\Sigma}\quad \rotatebox[origin=c]{270}{$\curvearrowright$}\;\; \til{H} $$ 
as follows (recall that $i$ is the canonical inclusion).
\begin{itemize} 
\item 
We define the chain homotopy    $\til{H}\colon \til{\Omega}^n(M,\Sigma)\ra \til{\Omega}^0(M,\Sigma)$ as the integral operator with kernel $\chi$, and the chain homotopy
$H\colon \Omega^n(M)\ra \Omega^0(M)$ as the restriction of $\til{H}$ to smooth forms. 
\item We define the operator $p\colon \til\Omega^n(M,\Sigma)\ra \Omega^n(M) $ as the integral operator with kernel $\delta^{(n,0)}_\mr{Diag}-(d*^{-1}d)_1\chi$. 
\item We define $p\colon  \til\Omega^0(M,\Sigma)\ra \Omega^0(M) $  as the integral operator with kernel $\delta^{(0,n)}_\mr{Diag}-(d*^{-1}d)_2\chi$.
\end{itemize}


With this construction, Theorem \ref{thm main} holds.
Indeed, relation (\ref{d_Q i = i d_Q}) is obvious. Relation (\ref{d_Q p = p d_Q}) follows from the identity at the level of integral kernels of left and right side:
$$(d*^{-1}d)_1 (\delta_{\mr{Diag}}^{(0,n)}-(d*^{-1}d)_2\chi)=(d*^{-1}d)_2 (\delta_{\mr{Diag}}^{(n,0)}-(d*^{-1}d)_1 \chi) . $$ 
Relations (\ref{chain homotopy rel H}), (\ref{chain homotopy rel tilH}) follow from an obvious identity at the level of integral kernels:
\begin{multline*}
(d*^{-1}d)_1 \chi+(d*^{-1}d)_2 \chi =\\ =\left(\delta^{(n,0)}_\mr{Diag}+\delta^{(0,n)}_\mr{Diag}\right)-\left(\delta^{(n,0)}_\mr{Diag}-(d*^{-1}d)_1\chi\right)-\left(\delta^{(0,n)}_\mr{Diag}-(d*^{-1}d)_2\chi\right) .
\end{multline*}
Locality is true by construction.

\section{$\p$-form electrodynamics in the first-order formalism}
\label{s: p-form ED}
Consider theory of $\p$-form field, with $\p\geq 0$, defined in the first order formalism by the classical action 
$$S^\mr{cl}_M=\int_M B\wedge  dA+\frac12 B\wedge *B$$ 
with $(A,B)\in \Omega^p(M)\oplus \Omega^{n-p-1}(M)$. In particular, for $\p=0$ this is the free massless scalar field; for $\p=1$ this is Maxwell's theory of electromagnetic field.

In BV formalism, the BV action is 
\begin{multline} \label{S p-form}
S_M=\int_M B \wedge dA +  A^+ \wedge dc_1 +\sum_{k=2}^\p  c^+_{k-1} \wedge d c_k + \frac12 B\wedge *B = \\
=\int_M \mc{B}\wedge d\mc{A}+\frac12 B\wedge *B ,
\end{multline}
where 
\begin{multline*} 
(\mc{A}=c_\p+\cdots+c_1+A+B^+ \quad ,\quad \mc{B}=B+A^++c_1^++\cdots +c_\p^+) \quad \in \\
\in \quad  \Omega^{0\cdots \p+1}(M)[\p]\oplus \Omega^{n-\p-1\cdots n}[n-\p-1] =\colon \FF_M .
\end{multline*}
Here $\Omega^{i\cdots j}(M)=\oplus_{k=i}^j\Omega^k(M)$ is the truncated de Rham complex.

The action (\ref{S p-form}) is the free theory action corresponding to the Poincar\'e cochain complex structure on $\FF_M$
with differential $d_Q$  given by 
$$ \hspace{-0.5cm}
\begin{array}{*{15}c}
\underset{c_\p}{\Omega^0} & \xra{d} &\cdots &\xra{d}& \underset{c_1}{\Omega^{\p-1}} & \xra{d} & \underset{A}{\Omega^\p} & \xra{d} & \underset{B^+}{\Omega^{\p+1} }
&&&&&& \\
&&&&&&& {\hspace{-0.5cm}\nearrow\hspace{-0.7cm} *} &&&&&&& \\
&&&&&&\underset{B}{\Omega^{n-\p-1}}&\xra{d}& \underset{A^+}{\Omega^{n-\p}} & \xra{d} & \underset{c_1^+}{\Omega^{n-\p+1}}& \xra{d} &\cdots & \xra{d} & \underset{c_\p^+}{\Omega^n}
\end{array}
$$
In particular, $d_Q(B)=dB+*B$. Dually, at the level of coordinates on the space of fields, one has a more familiar formula $Q_M(B^+)=d A + * B$. The constant $(-1)$-symplectic structure on $\FF_M$ is $\omega_M={(-1)^n\int_M \delta \mc{B}\wedge \delta\mc{A}}$. 

The phase space that BV-BFV formalism assigns to a closed $(n-1)$-manifold $\Sigma$ is
\begin{multline}\label{p-form theory phase space}
\Phi_\Sigma=\Omega^{0\cdots \p}(\Sigma)\oplus \Omega^{n-\p-1\cdots n-1}(\Sigma)\quad \ni\\ 
\ni \quad (\mc{A}_\Sigma=c_\p+\cdots+c_1+A\quad,\quad \mc{B}_\Sigma=B+A^++c_1^++\cdots+c_{\p-1}^+)  
\end{multline}
with differential $d_Q$ just the de Rham differential on both summands and the constant $0$-symplectic form $\int_\Sigma \delta \mc{B}_\Sigma\wedge \delta \mc{A}_\Sigma$.

In the setup of Conjecture \ref{claim: lin case}, with $n$-manifold $M$ cut into parts $M_1$, $M_2$ by a closed $(n-1)$-submanifold $\Sigma$, we have the canonical inclusion
\begin{equation}\label{p-form i}
\FF_M\xra{i} \til{\FF}_{M,\Sigma}=\FF_{M_1}\times_{\Phi_\Sigma}\FF_{M_2} .
\end{equation}

\subsection{Construction 1: Hodge homotopy}
For the map 
$p\colon \til{\FF}_{M,\Sigma}\ra \FF_M$, we take 
\begin{equation}\label{p-form p Hodge}
p=e^{-\epsilon \Delta},
\end{equation}
as in Section \ref{ss: scalar field Hodge homotopy}, but now acting on all forms of various degrees in the complex $\til{\FF}_{M,\Sigma}$. It is easy to see that $p$ is a chain map w.r.t. $d_Q$.

We define the chain homotopy $H\;\;\rotatebox[origin=c]{90}{$\curvearrowright$} \;\;\FF_M$ as
\begin{equation}\label{p-form H Hodge}
\vcenter{\hbox{ \scalebox{0.85}{\begin{picture}(0,0)%
\includegraphics{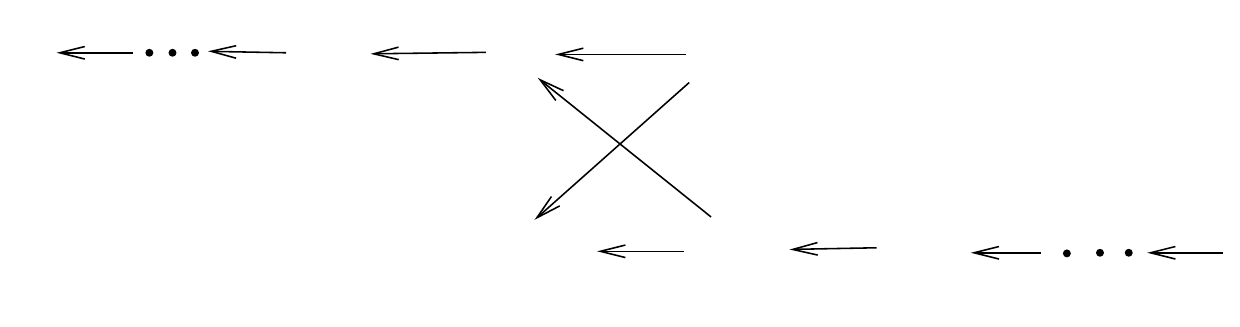}%
\end{picture}%
%
%
\setlength{\unitlength}{3947sp}%
\begingroup\makeatletter\ifx\SetFigFont\undefined%
\gdef\SetFigFont#1#2#3#4#5{%
  \reset@font\fontsize{#1}{#2pt}%
  \fontfamily{#3}\fontseries{#4}\fontshape{#5}%
  \selectfont}%
\fi\endgroup%
\begin{picture}(5940,1556)(691,-2334)
\put(3976,-1666){\makebox(0,0)[lb]{\smash{{\SetFigFont{12}{14.4}{\rmdefault}{\mddefault}{\updefault}{\color[rgb]{0,0,0}$\Psi$}%
}}}}
\put(4017,-1051){\makebox(0,0)[lb]{\smash{{\SetFigFont{12}{14.4}{\rmdefault}{\mddefault}{\updefault}{\color[rgb]{0,0,0}$\Omega^{\p+1}$}%
}}}}
\put(3072,-1051){\makebox(0,0)[lb]{\smash{{\SetFigFont{12}{14.4}{\rmdefault}{\mddefault}{\updefault}{\color[rgb]{0,0,0}$\Omega^\p$}%
}}}}
\put(2127,-1051){\makebox(0,0)[lb]{\smash{{\SetFigFont{12}{14.4}{\rmdefault}{\mddefault}{\updefault}{\color[rgb]{0,0,0}$\Omega^{\p-1}$}%
}}}}
\put(706,-1044){\makebox(0,0)[lb]{\smash{{\SetFigFont{12}{14.4}{\rmdefault}{\mddefault}{\updefault}{\color[rgb]{0,0,0}$\Omega^0$}%
}}}}
\put(3072,-1996){\makebox(0,0)[lb]{\smash{{\SetFigFont{12}{14.4}{\rmdefault}{\mddefault}{\updefault}{\color[rgb]{0,0,0}$\Omega^{n-\p-1}$}%
}}}}
\put(4017,-1996){\makebox(0,0)[lb]{\smash{{\SetFigFont{12}{14.4}{\rmdefault}{\mddefault}{\updefault}{\color[rgb]{0,0,0}$\Omega^{n-\p}$}%
}}}}
\put(4962,-1996){\makebox(0,0)[lb]{\smash{{\SetFigFont{12}{14.4}{\rmdefault}{\mddefault}{\updefault}{\color[rgb]{0,0,0}$\Omega^{n-\p+1}$}%
}}}}
\put(6616,-1996){\makebox(0,0)[lb]{\smash{{\SetFigFont{12}{14.4}{\rmdefault}{\mddefault}{\updefault}{\color[rgb]{0,0,0}$\Omega^n$}%
}}}}
\put(3609,-961){\makebox(0,0)[lb]{\smash{{\SetFigFont{12}{14.4}{\rmdefault}{\mddefault}{\updefault}{\color[rgb]{0,0,0}$\Xi$}%
}}}}
\put(2717,-954){\makebox(0,0)[lb]{\smash{{\SetFigFont{12}{14.4}{\rmdefault}{\mddefault}{\updefault}{\color[rgb]{0,0,0}$\Xi$}%
}}}}
\put(1877,-954){\makebox(0,0)[lb]{\smash{{\SetFigFont{12}{14.4}{\rmdefault}{\mddefault}{\updefault}{\color[rgb]{0,0,0}$\Xi$}%
}}}}
\put(1149,-954){\makebox(0,0)[lb]{\smash{{\SetFigFont{12}{14.4}{\rmdefault}{\mddefault}{\updefault}{\color[rgb]{0,0,0}$\Xi$}%
}}}}
\put(3707,-2243){\makebox(0,0)[lb]{\smash{{\SetFigFont{12}{14.4}{\rmdefault}{\mddefault}{\updefault}{\color[rgb]{0,0,0}$\Xi$}%
}}}}
\put(4659,-2236){\makebox(0,0)[lb]{\smash{{\SetFigFont{12}{14.4}{\rmdefault}{\mddefault}{\updefault}{\color[rgb]{0,0,0}$\Xi$}%
}}}}
\put(5529,-2236){\makebox(0,0)[lb]{\smash{{\SetFigFont{12}{14.4}{\rmdefault}{\mddefault}{\updefault}{\color[rgb]{0,0,0}$\Xi$}%
}}}}
\put(6384,-2259){\makebox(0,0)[lb]{\smash{{\SetFigFont{12}{14.4}{\rmdefault}{\mddefault}{\updefault}{\color[rgb]{0,0,0}$\Xi$}%
}}}}
\put(3241,-1659){\makebox(0,0)[lb]{\smash{{\SetFigFont{12}{14.4}{\rmdefault}{\mddefault}{\updefault}{\color[rgb]{0,0,0}$\Phi$}%
}}}}
\end{picture}%
} }} 
\end{equation}
where
\begin{equation}\label{p-form Xi Phi Psi}
\begin{gathered}
\Xi=(\mr{id}-e^{-\epsilon\Delta})\, d^* \Delta^{-1} , \\
\Phi = 
(-1)^{n+\p+1} d*^{-1}\Delta^{-1}(\mr{id}-e^{-\epsilon\Delta}) \, d
,\\ \Psi=(-1)^{n+\p+1}(\mr{id}-e^{-\epsilon \Delta})\Delta^{-1}* .
\end{gathered}
\end{equation}
The chain homotopy  $\til{H}\;\;\rotatebox[origin=c]{90}{$\curvearrowright$} \;\;\til\FF_{M,\Sigma}$ is defined by the same formulae, where now the operators (\ref{p-form Xi Phi Psi}) are understood as acting on forms with admissible singularities at $\Sigma$ constituting the fiber product $\til\FF_{M,\Sigma}$.

\begin{lemma}
Relations (\ref{d_Q i = i d_Q}), (\ref{d_Q p = p d_Q}), (\ref{chain homotopy rel H}), (\ref{chain homotopy rel tilH}) hold.
\end{lemma}
\begin{proof}
Relation (\ref{d_Q i = i d_Q}) is obvious. Let us check (\ref{chain homotopy rel H}). On
$\Omega^{i}$ with $i\leq \p-1$ or $i\geq n-\p+1$, we have
$$
d_Q H+H d_Q + p\, i  = d \Xi + \Xi d+ e^{-\epsilon\Delta}= (\mr{id}-e^{-\epsilon \Delta})(dd^*+d^*d)\Delta^{-1}+e^{-\epsilon \Delta}=\mr{id} .
$$
On $\Omega^\p\oplus \Omega^{n-\p-1}$, we have
\begin{multline*}
(d_Q H+H d_Q + p\, i )(A, B) = 
\\
=(\underbrace{(d \Xi +\Xi d +e^{-\epsilon\Delta})}_{\mr{id}}A + \underbrace{(\Xi*+\Psi d)}_{0}B,  \underbrace{(\Xi d+\Phi*+e^{-\epsilon\Delta})}_{\mr{id}}B +\underbrace{\Phi d}_{0} A) \\= (A,B).
\end{multline*}
On $\Omega^{\p+1}\oplus \Omega^{n-\p}$, we have
\begin{multline*}
(d_Q H+H d_Q + p\, i )(B^+, A^+) = \\ =
( \underbrace{(d\Xi+*\Phi+e^{-\epsilon \Delta})}_{\mr{id}}B^+ 
+ \underbrace{(*\Xi+d\Psi)}_0 A^+ ,  \underbrace{(d\Xi+\Xi d+e^{-\epsilon \Delta})}_{\mr{id}} A^+ + \underbrace{d\Phi}_0 B^+) \\= (B^+,A^+).
\end{multline*}
This finishes the proof of (\ref{chain homotopy rel H}). Relation (\ref{chain homotopy rel tilH}) is checked by the same computation (at the level of $\til\Omega^\bullet$).
Chain map property of $p$, (\ref{d_Q p = p d_Q}) in fact follows from (\ref{chain homotopy rel tilH}): $p=\mr{id}-[d_Q, \til{H}]$ implies (since $d_Q^2=0$) ${[d_Q,p]=0}$.
\end{proof}

\begin{remark} One can construct the chain homotopy (\ref{p-form H Hodge}) together with the map (\ref{p-form p Hodge}) from topological quantum mechanics (TQM) in the sense of \cite{Losev} with evolution operator 
$$\mc{U}(t,dt)=e^{-t[d_Q,\G] - dt\, \G }\quad  \in \Omega^\bt(\mathbb{R}_+)\otimes \mr{End}(\F)$$
-- a $(d_t+\mr{ad}_{d_Q})$-closed operator-valued form on $\mathbb{R}_+$, where the degree $-1$ operator $\mathbb{G}$ (the ``non-normalized homotopy'') is chosen to be:
$$\G\colon \vcenter{\hbox{\scalebox{0.85}{ \begin{picture}(0,0)%
\includegraphics{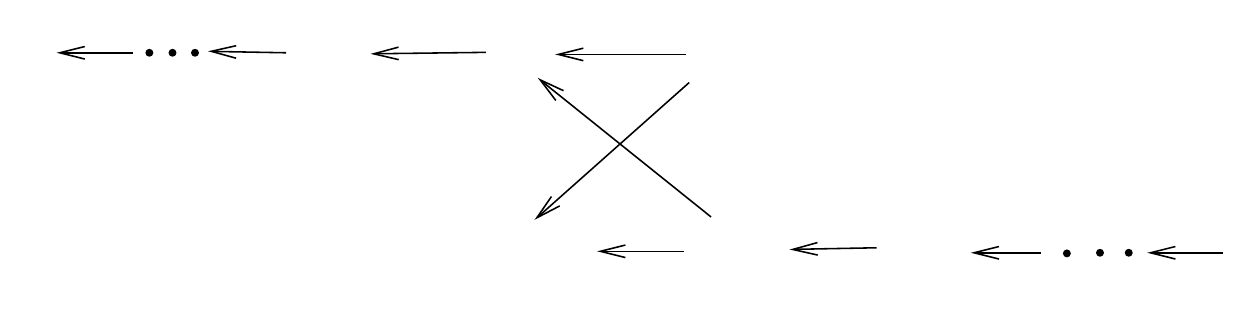}%
\end{picture}%
%
%
\setlength{\unitlength}{3947sp}%
\begingroup\makeatletter\ifx\SetFigFont\undefined%
\gdef\SetFigFont#1#2#3#4#5{%
  \reset@font\fontsize{#1}{#2pt}%
  \fontfamily{#3}\fontseries{#4}\fontshape{#5}%
  \selectfont}%
\fi\endgroup%
\begin{picture}(5940,1556)(691,-2334)
\put(3976,-1666){\makebox(0,0)[lb]{\smash{{\SetFigFont{12}{14.4}{\rmdefault}{\mddefault}{\updefault}{\color[rgb]{0,0,0}$\pm *$}%
}}}}
\put(4017,-1051){\makebox(0,0)[lb]{\smash{{\SetFigFont{12}{14.4}{\rmdefault}{\mddefault}{\updefault}{\color[rgb]{0,0,0}$\Omega^{\p+1}$}%
}}}}
\put(3072,-1051){\makebox(0,0)[lb]{\smash{{\SetFigFont{12}{14.4}{\rmdefault}{\mddefault}{\updefault}{\color[rgb]{0,0,0}$\Omega^\p$}%
}}}}
\put(2127,-1051){\makebox(0,0)[lb]{\smash{{\SetFigFont{12}{14.4}{\rmdefault}{\mddefault}{\updefault}{\color[rgb]{0,0,0}$\Omega^{\p-1}$}%
}}}}
\put(706,-1044){\makebox(0,0)[lb]{\smash{{\SetFigFont{12}{14.4}{\rmdefault}{\mddefault}{\updefault}{\color[rgb]{0,0,0}$\Omega^0$}%
}}}}
\put(3072,-1996){\makebox(0,0)[lb]{\smash{{\SetFigFont{12}{14.4}{\rmdefault}{\mddefault}{\updefault}{\color[rgb]{0,0,0}$\Omega^{n-\p-1}$}%
}}}}
\put(4017,-1996){\makebox(0,0)[lb]{\smash{{\SetFigFont{12}{14.4}{\rmdefault}{\mddefault}{\updefault}{\color[rgb]{0,0,0}$\Omega^{n-\p}$}%
}}}}
\put(4962,-1996){\makebox(0,0)[lb]{\smash{{\SetFigFont{12}{14.4}{\rmdefault}{\mddefault}{\updefault}{\color[rgb]{0,0,0}$\Omega^{n-\p+1}$}%
}}}}
\put(6616,-1996){\makebox(0,0)[lb]{\smash{{\SetFigFont{12}{14.4}{\rmdefault}{\mddefault}{\updefault}{\color[rgb]{0,0,0}$\Omega^n$}%
}}}}
\put(3609,-961){\makebox(0,0)[lb]{\smash{{\SetFigFont{12}{14.4}{\rmdefault}{\mddefault}{\updefault}{\color[rgb]{0,0,0}$d^*$}%
}}}}
\put(2717,-954){\makebox(0,0)[lb]{\smash{{\SetFigFont{12}{14.4}{\rmdefault}{\mddefault}{\updefault}{\color[rgb]{0,0,0}$d^*$}%
}}}}
\put(1877,-954){\makebox(0,0)[lb]{\smash{{\SetFigFont{12}{14.4}{\rmdefault}{\mddefault}{\updefault}{\color[rgb]{0,0,0}$d^*$}%
}}}}
\put(1149,-954){\makebox(0,0)[lb]{\smash{{\SetFigFont{12}{14.4}{\rmdefault}{\mddefault}{\updefault}{\color[rgb]{0,0,0}$d^*$}%
}}}}
\put(3707,-2243){\makebox(0,0)[lb]{\smash{{\SetFigFont{12}{14.4}{\rmdefault}{\mddefault}{\updefault}{\color[rgb]{0,0,0}$d^*$}%
}}}}
\put(4659,-2236){\makebox(0,0)[lb]{\smash{{\SetFigFont{12}{14.4}{\rmdefault}{\mddefault}{\updefault}{\color[rgb]{0,0,0}$d^*$}%
}}}}
\put(5529,-2236){\makebox(0,0)[lb]{\smash{{\SetFigFont{12}{14.4}{\rmdefault}{\mddefault}{\updefault}{\color[rgb]{0,0,0}$d^*$}%
}}}}
\put(6384,-2259){\makebox(0,0)[lb]{\smash{{\SetFigFont{12}{14.4}{\rmdefault}{\mddefault}{\updefault}{\color[rgb]{0,0,0}$d^*$}%
}}}}
\put(3241,-1659){\rotatebox{40.0}{\makebox(0,0)[lb]{\smash{{\SetFigFont{12}{14.4}{\rmdefault}{\mddefault}{\updefault}{\color[rgb]{0,0,0}$\pm d*d$}%
}}}}}
\end{picture}%
 }}} $$
Note that with this choice the Hamiltonian of the TQM $\HHH\colon=[d_Q,\G]$ is just the Laplace operator $\Delta$.
Then, one has $$p\,i=\mc{U}|_{t=\epsilon}\quad = e^{-t\HHH}$$ and 
$$H=-\int_0^\epsilon \mc{U}\quad = \G\HHH^{-1}(1-e^{-\epsilon \HHH})$$
 (where we integrate the 1-form component of $\mc{U}$ along $\mathbb{R}_+$).  The relation (\ref{chain homotopy rel H}) follows by Stokes' theorem for the integral over $[0,\epsilon]$, using closedness of $\mc{U}$.

We remark that $\G$ does not square to zero.

\end{remark}

\subsection{Construction 2: homotopy local near the interface}
\label{ss: p-form 1st order local homotopy}
As in Section \ref{ss: scalar field local homotopy}, fix an open neighborhood $U\subset M$ of $\Sigma$. Fix a function $\mu$ on $M\times M$ as in the proof of Lemma \ref{lemma 4.2}. Let $G\in\bigoplus_{i=0}^n\Omega^{i,n-i}_{C^0}(U\times U)$ be the integral kernel of the inverse of the Laplace operator $\Delta_U$ acting on forms on $U$ with Dirichlet boundary conditions at $\dd U$; $G$ is smooth on $U\times U$ away from the diagonal, where it is only continuous (hence the subscript $C^0$, for continuous forms). 

Let $\Gamma\in \bigoplus_{i=0}^n\Omega^{i,n-i}_{C^0}(M\times M)$ be the form on $M\times M$ defined as
$$ \Gamma= 
\left\{ 
\begin{array}{ccc}
\mu\cdot G & \mr{in} & U\times U, \\
0 & \mr{outside} & U\times U
\end{array}
\right.
 $$
Let $\Delta^{-1}_\mu\colon \Omega^\bt(M)\ra \Omega^\bt(M)$ be the operator with integral kernel $\Gamma$.\footnote{We are thinking of it as a modification of $\Delta_U^{-1}$ provided by the bump function $\mu$, hence the notation; we are not thinking of $\Delta_\mu^{-1}$ as an inverse of some operator $\Delta_\mu$.} We remark that $\Delta^{-1}_\mu$ commutes with the Hodge star (since $*_1\Gamma=\mu *_1G=\mu *_2 G=*_2\Gamma$ where $*_1G=*_2G$ follows from the fact that $\Delta_U$ -- and its inverse -- commutes with $*$), but generally does not commute with $d$.


We will first define the chain homotopy operators and then, using them, construct the map $p\colon \til\FF_{M,\Sigma}\ra \FF_M$.

We define the chain homotopy $H\;\;\rotatebox[origin=c]{90}{$\curvearrowright$} \;\;\FF_M$ as
\begin{equation}\label{p-form H}
\vcenter{\hbox{ 
\scalebox{0.85}{
\begin{picture}(0,0)%
\includegraphics{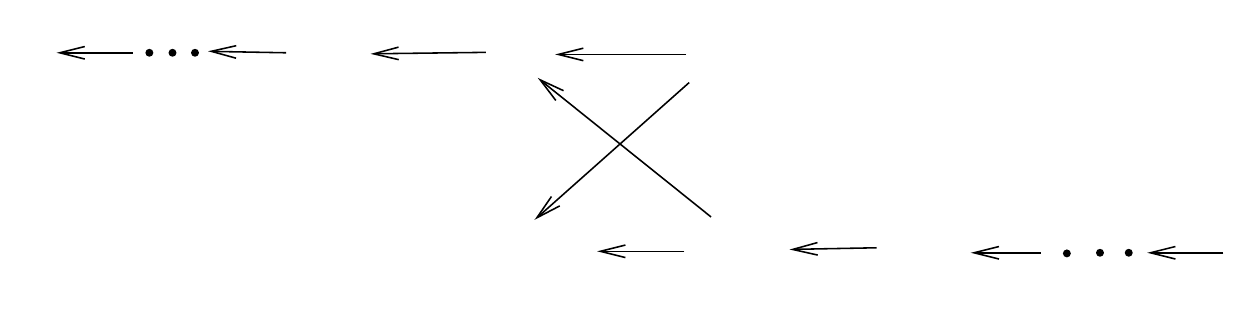}%
\end{picture}%
%
%
\setlength{\unitlength}{3947sp}%
\begingroup\makeatletter\ifx\SetFigFont\undefined%
\gdef\SetFigFont#1#2#3#4#5{%
  \reset@font\fontsize{#1}{#2pt}%
  \fontfamily{#3}\fontseries{#4}\fontshape{#5}%
  \selectfont}%
\fi\endgroup%
\begin{picture}(5940,1556)(691,-2334)
\put(3976,-1666){\makebox(0,0)[lb]{\smash{{\SetFigFont{12}{14.4}{\rmdefault}{\mddefault}{\updefault}{\color[rgb]{0,0,0}$\Psi$}%
}}}}
\put(4017,-1051){\makebox(0,0)[lb]{\smash{{\SetFigFont{12}{14.4}{\rmdefault}{\mddefault}{\updefault}{\color[rgb]{0,0,0}$\Omega^{\p+1}$}%
}}}}
\put(3072,-1051){\makebox(0,0)[lb]{\smash{{\SetFigFont{12}{14.4}{\rmdefault}{\mddefault}{\updefault}{\color[rgb]{0,0,0}$\Omega^\p$}%
}}}}
\put(2127,-1051){\makebox(0,0)[lb]{\smash{{\SetFigFont{12}{14.4}{\rmdefault}{\mddefault}{\updefault}{\color[rgb]{0,0,0}$\Omega^{\p-1}$}%
}}}}
\put(706,-1044){\makebox(0,0)[lb]{\smash{{\SetFigFont{12}{14.4}{\rmdefault}{\mddefault}{\updefault}{\color[rgb]{0,0,0}$\Omega^0$}%
}}}}
\put(3072,-1996){\makebox(0,0)[lb]{\smash{{\SetFigFont{12}{14.4}{\rmdefault}{\mddefault}{\updefault}{\color[rgb]{0,0,0}$\Omega^{n-\p-1}$}%
}}}}
\put(4017,-1996){\makebox(0,0)[lb]{\smash{{\SetFigFont{12}{14.4}{\rmdefault}{\mddefault}{\updefault}{\color[rgb]{0,0,0}$\Omega^{n-\p}$}%
}}}}
\put(4962,-1996){\makebox(0,0)[lb]{\smash{{\SetFigFont{12}{14.4}{\rmdefault}{\mddefault}{\updefault}{\color[rgb]{0,0,0}$\Omega^{n-\p+1}$}%
}}}}
\put(6616,-1996){\makebox(0,0)[lb]{\smash{{\SetFigFont{12}{14.4}{\rmdefault}{\mddefault}{\updefault}{\color[rgb]{0,0,0}$\Omega^n$}%
}}}}
\put(3609,-961){\makebox(0,0)[lb]{\smash{{\SetFigFont{12}{14.4}{\rmdefault}{\mddefault}{\updefault}{\color[rgb]{0,0,0}$\Xi$}%
}}}}
\put(2717,-954){\makebox(0,0)[lb]{\smash{{\SetFigFont{12}{14.4}{\rmdefault}{\mddefault}{\updefault}{\color[rgb]{0,0,0}$\Xi$}%
}}}}
\put(1877,-954){\makebox(0,0)[lb]{\smash{{\SetFigFont{12}{14.4}{\rmdefault}{\mddefault}{\updefault}{\color[rgb]{0,0,0}$\Xi$}%
}}}}
\put(1149,-954){\makebox(0,0)[lb]{\smash{{\SetFigFont{12}{14.4}{\rmdefault}{\mddefault}{\updefault}{\color[rgb]{0,0,0}$\Xi$}%
}}}}
\put(3707,-2243){\makebox(0,0)[lb]{\smash{{\SetFigFont{12}{14.4}{\rmdefault}{\mddefault}{\updefault}{\color[rgb]{0,0,0}$\Xi'$}%
}}}}
\put(4659,-2236){\makebox(0,0)[lb]{\smash{{\SetFigFont{12}{14.4}{\rmdefault}{\mddefault}{\updefault}{\color[rgb]{0,0,0}$\Xi'$}%
}}}}
\put(5529,-2236){\makebox(0,0)[lb]{\smash{{\SetFigFont{12}{14.4}{\rmdefault}{\mddefault}{\updefault}{\color[rgb]{0,0,0}$\Xi'$}%
}}}}
\put(6384,-2259){\makebox(0,0)[lb]{\smash{{\SetFigFont{12}{14.4}{\rmdefault}{\mddefault}{\updefault}{\color[rgb]{0,0,0}$\Xi'$}%
}}}}
\put(3241,-1659){\makebox(0,0)[lb]{\smash{{\SetFigFont{12}{14.4}{\rmdefault}{\mddefault}{\updefault}{\color[rgb]{0,0,0}$\Phi$}%
}}}}
\end{picture}%
 
}
}} 
\end{equation}
where
\begin{equation}\label{components of H local}
\begin{gathered}
\Xi=\Delta_\mu^{-1} d^* ,\; \Xi'=d^*\Delta_\mu^{-1},\\ 
\Phi=
(-1)^{n+\p+1} d*^{-1}\Delta_\mu^{-1} d
,\; \Psi=(-1)^{n+\p+1}*\Delta_\mu^{-1} .
\end{gathered}
\end{equation}

The chain homotopy $\til{H}\;\;\rotatebox[origin=c]{90}{$\curvearrowright$} \;\;\til\FF_{M,\Sigma}$ is defined by  the operators with the same integral kernels as in $H$, now acting on forms constituting $\til\FF_{M,\Sigma}$.

We construct the map $p\colon \til\FF_{M,\Sigma}\ra \FF_M$ as 
$p=\mr{id}-d_Q \til{H}-\til{H} d_Q$ (here the important point is that the image of $p$ is in smooth forms $\FF_M\subset \til\FF_{M,\Sigma}$). Explicitly, $p$ looks as follows:
\begin{equation}\label{components of p, p-form theory}
p=\left\{
\begin{array}{ccccc}
\mr{id}-d\Delta_\mu^{-1}d^*- \Delta_{\mu}^{-1}d^*d & \mr{on} & \til\Omega^i & \mr{with} & i\leq \p,\\
\mr{id}-dd^*\Delta_\mu^{-1}- d^*\Delta_{\mu}^{-1}d & \mr{on} & \til\Omega^i & \mr{with} & i\geq n-\p, \\
\mr{id}-d\Delta_\mu^{-1}d^*- d^*\Delta_{\mu}^{-1}d & \mr{on} & \til\Omega^{i} & \mr{with}  & i\in \{\p+1,n-\p-1\} 
\end{array}
\right.
\end{equation}
A note on why $\Phi$ and $\Psi$ components of the homotopy are crucial for the construction to work: here for the case $i=n-\p$, it is important that maps (\ref{components of H local}) satisfy $*\Xi'+d\Psi=0$; for $i=n-\p-1$, it is important that $\Xi *+\Psi d=0$; for $i=\p$ and $i=\p+1$, it is important that $d\Phi=\Phi d=0$. If not for these properties, $p$ would not be diagonal w.r.t. components of $\FF_M$, e.g., $p$ acting on $\til\Omega^\p$ could have a component  in $\til\Omega^{n-\p-1}$. Even more importantly, these off-diagonal components of $p$ would fail to map to smooth forms (e.g. if one would try to set $\Phi$ or $\Psi$ to zero).

The integral kernel of each component of $p$ is:
\begin{itemize}
\item zero in $V\subset U\times U$ -- a neighborhood of $\mr{Diag}_U$ where $\mu=1$ (see the proof of Lemma \ref{lemma 4.2}),\footnote{Indeed, in $V$ the kernels of $\Delta^{-1}_\mu$ and of the true Green function $\Delta^{-1}_U$ coincide and $\Delta^{-1}_U$ commutes both with $d$ and with $*$, thus, by inspecting (\ref{components of p, p-form theory}), the kernel of each component vanishes in $V$.}
\item smooth in $U\times U$ (smoothness in $U\times U\backslash V$ is obvious and in $V$ the kernel of $p$ is smooth by the previous point), 
\item $\delta_\mr{Diag}$ outside $U\times U$.
\end{itemize}
These properties in particular imply the claim made above, that $p$ takes forms in $\til{\FF}_{M,\Sigma}$ in smooth forms in $\FF_M$.

This finishes the proof of Theorem \ref{thm main} in the present example.

\begin{remark} One can construct an $(i,p,H,\til{H})$-package for abelian Chern--Simons (or abelian $BF$) -- as an alternative to the construction of Section \ref{s: ex - abCS} -- from the formulae above:
\begin{itemize}
\item (Hodge-theoretic version.) One can set $p=e^{-\epsilon \Delta}$ and homotopies $H$ and $\til{H}$ to be given by the operator $\Xi={(\mr{id}-e^{-\epsilon\Delta})\, d^* \Delta^{-1}}$ from (\ref{p-form Xi Phi Psi}).
\item (Local-near-$\Sigma$ version.) One can set $H$ and $\til{H}$ to be given by $\Xi=\Delta_\mu^{-1} d^*$ from (\ref{components of H local}) and $p=\mr{id}
-d\Xi-\Xi d$.\footnote{Alternatively, one may replace $\Xi=\Delta_\mu^{-1} d^*$ with $\Xi'=d^*\Delta_\mu^{-1}$ from (\ref{components of H local}) in this construction.}
\end{itemize}
\end{remark}

\subsection{Second-order formalism}\label{ss: p-form 2nd order}
The same model---$\p$-form electrodynamics---cast in the second-order formalism has BV action 
$$S_M=\int_M \frac12 dA\wedge *dA + A^+\wedge dc_1+\sum_{k=2}^p c^+_{k-1}\wedge dc_k . $$
The Poincar\'e complex of fields $\FF_M$ is\footnote{\label{footnote: d*d sign in second order p-form theory}
A note on signs: to explain why we write $d*^{-1}d$ rather than $d*d$ as a term in $d_Q$ (and also why $*^{-1}dA$ appears in $\pi$ rather than $*dA$), one needs to look at the full BV-BFV package for the second-order $\p$-form theory. With $S_M$ as above, $\omega_M=(-1)^n\int_M \delta \bar{\mc{B}}\wedge \delta\bar{\mc{A}}$, $Q_M=\int_M d\bar{\mc{A}}\frac{\delta}{\delta\bar{\mc{A}}}+(d\bar{\mc{B}}+d*^{-1}d A)\frac{\delta}{\delta\bar{\mc{B}}}$ and $\alpha_{\dd M}=(-1)^{n-1}\int_{\dd M}\bar{\mc{B}}\wedge\delta \bar{\mc{A}}+(*^{-1}d A)\wedge\delta A$, one has the main BV-BFV relation (\ref{BV-BFV eq}). Here we denoted $\bar{\mc{A}}=c_\p+\cdots+c_1+A$, $\bar{\mc{B}}=A^++c_1^++\cdots + c_\p^+$ the truncated bulk superfields.
} 
$$
\underset{c_\p}{\Omega^0}\xra{d}\cdots \xra{d}\underset{c_1}{\Omega^{\p-1}}\xra{d} \underset{A}{\Omega^\p} \xra{d*^{-1}d} \underset{A^+}{\Omega^{n-\p}}\xra{d} \underset{c_1^+}{\Omega^{n-\p+1}}\xra{d}\cdots \xra{d} \underset{c_\p^+}{\Omega^n}  .
$$
The phase space $\Phi_\Sigma$ is the same as above, (\ref{p-form theory phase space}); the projection $\FF_{M_{1,2}}\ra \Phi_\Sigma$ is the pullback of forms to $\Sigma$, except for $\p$-forms where $\pi(A)=(A|_\Sigma, (*^{-1}dA)|_\Sigma)$.

We again have the canonical inclusion $i\colon \FF_M\hra \til\FF_{M,\Sigma}={\FF_{M_1}\times_{\Phi_\Sigma}\FF_{M_2}}$ 
and we again have two variants of the rest of the $(i,p,H,\til{H})$ package.

\textbf{Hodge-theoretic version.} We set $p=e^{-\epsilon \Delta}\colon \til\FF_{M,\Sigma}\ra \FF_M$ and we set the chain homotopy $H$ to be
\begin{equation*}\label{p-form 2nd order H}
H:\qquad \Omega^0\xla{\Xi}\cdots \xla{\Xi}\Omega^{\p-1}\xla{\Xi} \Omega^\p \xla{-\Psi} \Omega^{n-\p}\xla{\Xi} \Omega^{n-\p+1}\xla{\Xi}\cdots \xla{\Xi} \Omega^n,
\end{equation*}
with $\Psi$, $\Xi$ as in (\ref{p-form Xi Phi Psi}). Likewise, we set $\til{H}$ to be given by the same operators (i.e. defined by the same integral kernels) acting on forms in $\til\FF_{M,\Sigma}$.

\textbf{Local-near-$\Sigma$ version.} We set $H$ to be given by 
$$
H:\qquad \Omega^0\xla{\Xi}\cdots \xla{\Xi}\Omega^{\p-1}\xla{\Xi} \Omega^\p \xla{-\Psi} \Omega^{n-\p}\xla{\Xi'} \Omega^{n-\p+1}\xla{\Xi'}\cdots \xla{\Xi'} \Omega^n ,
$$
where now the component operators $\Psi,\Xi,\Xi'$ are as in (\ref{components of H local}); $\til{H}$ is again given by the same operators acting on $\til\Omega$-forms. We define $p$ by formulae (\ref{components of p, p-form theory}) where we forget the components in $\Omega^i$ with $i=\p+1$, $i=n-\p-1$, as those are absent (integrated out) in the second-order formulation of the $\p$-form theory.

In both versions, items (\ref{thm (a)}), (\ref{thm (b)}) of Theorem \ref{thm main} hold.  
Second version additionally satisfies item (\ref{thm (c)})  -- locality.

%

\end{document}